\documentclass[11pt,a4paper]{article}
\usepackage{fullpage}
\usepackage{graphicx}
\usepackage{clrscode}
\usepackage{algorithmic}
\usepackage{algorithm}
\usepackage{amsmath, amssymb}
\usepackage{euscript}
\usepackage{xspace}
\usepackage{paralist}
\usepackage{wrapfig}

\newtheorem{theorem}{Theorem}
\newtheorem{lemma}{Lemma}
\newtheorem{defn}{Definition}
\newtheorem{observation}{Observation}
\newtheorem{corollary}{Corollary}

\newenvironment{proof}{Proof:}{\qed}
\def\squareforqed{\hbox{\rlap{$\sqcap$}$\sqcup$}}
\def\qed{\ifmmode\squareforqed\else{\unskip\nobreak\hfil
\penalty50\hskip1em\null\nobreak\hfil\squareforqed
\parfillskip=0pt\finalhyphendemerits=0\endgraf}\fi}

\newcommand{\figlab}[1]{\label{fig:#1}}
\newcommand{\figref}[1]{Figure~\ref{fig:#1}}

\newcommand{\seclab}[1]{\label{sec:#1}}
\newcommand{\secref}[1]{Section~\ref{sec:#1}}

\newcommand{\obslab}[1]{\label{observation:#1}}
\newcommand{\obsref}[1]{Observation~\ref{observation:#1}}

\newcommand{\defref}[1]{Definition~\ref{def:#1}}
\providecommand{\deflab}[1]{\label{def:#1}}

\newcommand{\algref}[1]{Algorithm~\ref{alg:#1}}
\providecommand{\alglab}[1]{\label{alg:#1}}

\newcommand{\lemlab}[1]{\label{lem:#1}}
\newcommand{\lemref}[1]{Lemma~\ref{lem:#1}}

\newcommand{\Eqlab}[1]{\label{eq:#1}}
\newcommand{\Eqref}[1]{Equation~\ref{eq:#1}}

\newcommand{\thmref}[1]{Theorem~\ref{theo:#1}}
\newcommand{\thmlab}[1]{{\label{theo:#1}}}

\newcommand{\corlab}[1]{\label{cor:#1}}
\newcommand{\corref}[1]{Corollary~\ref{cor:#1}}

\newcommand {\heig} [4] {\langle #1, #2, #3, #4 \rangle}
\def\low(#1){\ensuremath{\id{low}(#1)}}
\def\high(#1){\ensuremath{\id{high}(#1)}}
\def\steep(#1,#2,#3){\ensuremath{\sigma_{#1}(#2,#3)}}
\def\elev(#1,#2){\ensuremath{\id{elev}_{#1}(#2)}}
\def\level(#1,#2){\ensuremath{\id{level}_{#1}(#2)}}
\def\level(#1){\ensuremath{\id{level}(#1)}}
\def\symb(#1){\ensuremath{\id{elev}^{*}(#1)}}
\def\hre(#1){\ensuremath{\id{maxReachable}(#1)}}
\def\tenthre(#1){\ensuremath{\id{reachable}[#1]}}
\def\flowsto#1{\vtop{\hbox{\,$\rightarrow$\,}\kern-2ex\hbox{\scriptsize\,$#1$}}}
\def\flowsvia#1{\vtop{\hbox{\scriptsize\,$~#1$}\kern-2ex\hbox{\,$\rightarrow$\,}}}
\def\avoids#1#2{\vbox{\hbox{\scriptsize$\setminus #1$}\kern-1.7ex\flowsto{#2}}}
\def\nflowsto#1{\vtop{\hbox{\,$\not\rightarrow$\,}\kern-1.7ex\hbox{\scriptsize\,$#1$}}}

\renewcommand{\Re}{{\rm I\!\hspace{-0.025em} R}}

\newcommand{\ReT}{\ensuremath{\EuScript{R}_{T}}\xspace}
\newcommand{\FP}{\ensuremath{\Pi}}
\newcommand{\FPstable}{\ensuremath{\widetilde{\Pi}}}

\newcommand{\Rlow}{\ensuremath{{R}^{-}}\xspace}
\newcommand{\Rhigh}{\ensuremath{{R}^{+}}\xspace}
\newcommand{\emphi}[1]{\textbf{\emph{#1}}}
\def\WS(#1,#2){\ensuremath{\EuScript{W}_{#1}(#2)}}
\def\PoWS(#1){\ensuremath{\EuScript{W}_{\cup}\left(#1\right)}}
\def\AvWS#1(#2){\ensuremath{\EuScript{W}_{\cup}^{\setminus #1}\left(#2\right)}}
\def\CoWS(#1){\ensuremath{\EuScript{W}_{\cap}(#1)}}
\def\PsWS(#1){\ensuremath{\EuScript{W}_{\rlap{\kern0.3ex$\cdot$}\cap}(#1)}}
\def\Del(#1){\ensuremath{\EuScript{D}(#1)}}
\def\PoDel(#1){\ensuremath{\EuScript{D}_{\cup}(#1)}}
\newcommand{\WSO}{\ensuremath{\overline{R}}\xspace}
\def\CanonR(#1){\ensuremath{R_{\cup}(#1)}}
\newcommand{\expandPWS}{\ensuremath{\proc{Expand}}\xspace}
\newcommand{\computePWS}{\ensuremath{\proc{PotentialWS}}\xspace}
\newcommand{\expandPWSX}[2]{\ensuremath{\expandPWS(#1,#2)}\xspace}
\newcommand{\computePWSX}[1]{\ensuremath{\computePWS(#1)}\xspace}
\newcommand{\expandPD}{\ensuremath{\proc{ExpandDown}}\xspace}

\newcommand{\expandPDX}[2]{\ensuremath{\expandPD(#1,#2)}\xspace}

\newcommand{\compl}[1]{\ensuremath{\left(#1\right)^c}}
\newcommand{\clearance}{\ensuremath{\id{bar}}}
\newcommand{\PotentialWS}{\PoWS}
\newcommand{\PersistentWS}{\PsWS}
\newcommand{\eps}{\varepsilon}
\def\crossing(#1,#2){\ensuremath{{\cal X}_{#1}(#2)}}
\def\PoBound(#1){\ensuremath{{\cal X}_{\cup}(#1)}}

\newcommand{\XSays}[3]{{
      {$\rule[-0.12cm]{0.2in}{0.5cm}$\fbox{\tt
            #1 #2:} }
      \small #3
      \marginpar{\tt #1}
      {$\rule[0.1cm]{0.3in}{0.1cm}$\fbox{\tt
            end}$\rule[0.1cm]{0.3in}{0.1cm}$}
   }
}
\newcommand{\XSaysExt}[3]{{
      ~\\
      \fbox{\begin{minipage}{0.99\linewidth}
             {$\rule[-0.12cm]{0.2in}{0.5cm}$\fbox{\tt
                   #1 #2:}} \small #3
         \end{minipage}}
      \marginpar{\tt #1}}}

\newcommand{\Anne}[2][says]{{\XSays{Anne}{#1}{#2}}}
\newcommand{\AnneX}[2][says]{{\XSaysExt{Anne}{#1}{#2}}}
\newcommand{\HH}[2][says]{{\XSays{HH}{#1}{#2}}}
\newcommand{\HHX}[2][says]{{\XSaysExt{HH}{#1}{#2}}}

\newcommand{\maarten}[2][says]{{\XSays{Maarten}{#1}{#2}}}
\newcommand{\maartenX}[2][says]{{\XSaysExt{Maarten}{#1}{#2}}}

\newcommand{\email}[1]{\texttt{#1}}

\begin{document}
\title{Flow Computations on Imprecise Terrains}

\author{%
   Anne Driemel%
   \thanks{
      Department of Information and Computing Sciences,
      Utrecht University, The Netherlands,
      \email{anne@cs.uu.nl};
      This work has been supported by the Netherlands Organisation for
      Scientific Research (NWO) under RIMGA (Realistic Input Models
      for Geographic Applications).} %
   \and%
   Herman Haverkort%
   \thanks{
   Department of Computer Science,
   Eindhoven University of Technology,
   the Netherlands,
   \email{cs.herman@haverkort.net}
   }
   \and%
   Maarten L\"{o}ffler%
   \thanks{
   Computer Science Department,
   University of California, Irvine,
   USA,
   \email{mloffler@uci.edu};
   Funded by the U.S. Office of Naval Research under grant N00014-08-1-1015.}
   \and%
   Rodrigo Silveira%
   \thanks{
   Departament de Matem\`atica Aplicada II,
   Universitat Polit\`ecnica de Catalunya,
   Spain,
   \email{rodrigo.silveira@upc.edu};
   Supported by the Netherlands Organisation for Scientific Research (NWO).}
}

\maketitle

%


\begin{abstract} We study water flow computation on imprecise
terrains.  We consider two approaches to modeling flow on a terrain: one where
water flows across the surface of a polyhedral terrain in the direction of
steepest descent, and one where water only flows along the edges of a predefined
graph, for example a grid or a triangulation.  In both cases each vertex has an
imprecise elevation, given by an interval of possible values, while its
$(x,y)$-coordinates are fixed.  For the first model, we show that the problem of
deciding whether one vertex may be contained in the watershed of another is
NP-hard.  In contrast, for the second model we give a simple $O(n \log n)$ time
algorithm to compute the minimal and the maximal watershed of a vertex, or a set
of vertices, where $n$ is the number of edges of the graph.  On a grid model, we
can compute the same in $O(n)$ time.
\end{abstract}

\begin{centering}
\textit{Rose knew almost everything that water can do,\\there are an awful lot when you think what.}

\medskip
{\small Gertrude Stein, \textit{The World is Round}.}

\end{centering}

\section{Introduction}

Simulating the flow of water on a terrain is a problem that has been studied
for a long time in geographic information science (\textsc{gis}), and has
received considerable attention from the computational geometry community due
to the underlying geometric problems~\cite{bbd-crtt-96, ms-ecw-99,
bch-cfft-10}. It can be an important tool in analyzing flash floods for risk
management~\cite{borga-08}, for stream flow
forecasting~\cite{nature-koster-10}, and in the general study of
geomorphological processes~\cite{nature-craddock-10}, and it could contribute
to obtaining more reliable climate change predictions~\cite{tetzlaff-08}.

When modeling the flow of water across a terrain, it is generally assumed that
water flows downward in the direction of steepest descent.  It is common
practice to compute drainage networks and catchment areas directly from a
digital elevation model of the terrain.
Most hydrological research in \textsc{gis} models the terrain surface with a grid
in which each cell can drain to one or more of its eight neighbors~(e.g.~\cite{tarboton1997new}).
This can also be modeled as a computation on a graph, in which each node
represents a grid cell and each edge represents the adjacency of two neighbors
in the grid.
Alternatively, one could use an irregular network in which each node drains to
one or more of its neighbors, which may reduce the required storage space by
allowing less interesting parts of the terrain to have a lower sampling density.
We will refer to this as
the \emph{network model}, and we assume that, from every node, water flows down
along the steepest incident edge.
Assuming the elevation data is exact, drainage networks can be computed
efficiently in this model (e.g.~\cite{Danner2007}).
In computational geometry and topology, researchers have studied flow path and
drainage network computations on triangulated polyhedral
surfaces~(e.g.~\cite{deberg-10,tsiro-11, ls-ftt-04}). In this model, which we
call the \emph{surface model}, the flow of water can be traced across the
surface of a triangle. This avoids creating certain artifacts that arise when
working with grid models. However, the computations on polyhedral surfaces may
be significantly more difficult than on network models~\cite{tsiro-11}.

Naturally, all computations based on terrain data are subject to various sources
of uncertainty, like measurement, interpolation, and numerical errors.
The \textsc{gis} community has long recognized the importance of dealing with
uncertainty explicitly, in particular for hydrological modeling.  A common
approach is to model the elevation at a point of the
terrain using stochastic methods~\cite{w-udem-07}. However, the models
available in the hydrology literature are unsatisfactory~\cite{buytaert-08,sivakumar-08,montanari-06}
and computationally expensive~\cite{vrugt-2005}. A~particular challenge is posed by the fact that
hydrological computations can be extremely sensitive to small elevation
errors~\cite{hebeler20094,lindsay08}.
While most of these studies have been done in the network model, we note that 
there also exists work on the behaviour of watersheds on noisy terrains in the
surface model by Haverkort and Tsirogiannis \cite{ht-font-11}. 

\begin{figure}[t]\center
\includegraphics[width=0.88\textwidth]{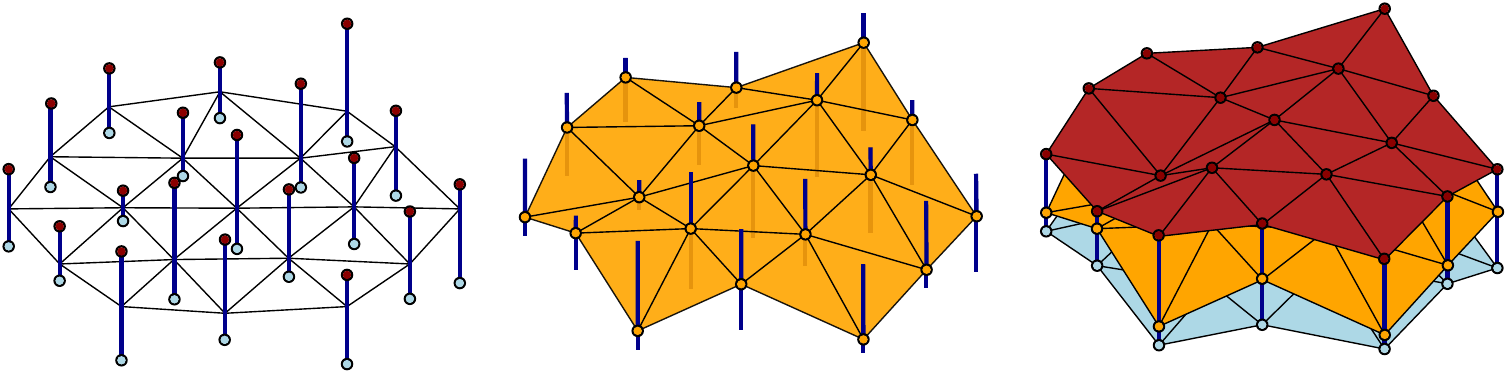}
\caption{
Left: An imprecise terrain. Each vertex of the triangulation has a elevation interval (gray).
Center: a realization of the imprecise terrain.
Right: the same realization together with the \emph{highest} and \emph{lowest} possible realizations of the imprecise terrain.
}
\label{fig:imprecise-terrain}
\end{figure}

A non-probabilistic model of imprecision that is often used in
computational geometry consists in representing an imprecise attribute (such as
location) by a region that is guaranteed to contain the true value.  This
approach has also been applied to polyhedral terrains~(e.g.~\cite{gls-smit-09,ke-opsp-07}),
replacing the exact elevation of each surface point by an imprecision interval
(see Figure~\ref{fig:imprecise-terrain}).
In this way, each terrain vertex does not have one fixed elevation, but a whole
range of possible elevations which includes the true elevation.
Choosing a concrete elevation for each vertex
results in a \emph{realization} of the imprecise terrain. The realization is a
(precise) polyhedral terrain.  Since the set of all possible realizations is
guaranteed to include the true (unknown) terrain, one can now obtain bounds on
parameters of the true terrain by computing the best- and worst-case values of
these parameters over the set of all possible realizations.
Note that we assume the error only in the $z$-coordinate (and not in the
$x,y$-coordinates). This is partially motivated by the fact that commercial
terrain data suppliers often only report elevation error~\cite{ft-ccedem-06}.
However, it is also a natural simplification of the model, since the true
terrain needs to have an elevation at any exact position in the plane.


In this paper we apply this model of imprecise terrains to problems related to
the simulation of water flow, both on terrains represented by surface models and
on terrains represented by network models.  One of the most fundamental
questions one may ask about water flow on terrains is whether water flows from a
point $p$ to another point $q$.
In the context of imprecise terrains, reasonable answers may be ``definitely
not'', ``possibly'', and ``definitely''.  The \emph{watershed} of a point in a
terrain is the part of the terrain that drains to this point.  Phrasing the same
question in terms of watersheds leads us to introduce the concepts of
\emph{potential} (maximal) and \emph{persistent} (minimal) watersheds.



\paragraph{Results}
In \secref{surface} we show that the problem of deciding whether water can flow
between two given points in the surface model is NP-hard.
%
Fortunately, the situation is much better for the network model, and therefore as a
special case also for the D-8 grid model which is widely adopted in \textsc{gis} applications.
In \secref{network} and \secref{regular} we present various results using this model.
In \secref{potential-ws} we present an algorithm to compute the potential
watershed of a point.  On a terrain with $n$ edges, our algorithm
runs in $O(n \log n)$ time; for grid models the running time can even be
improved to $O(n)$. We extend these techniques and achieve the same running times
for computing the potential downstream area of a point in \secref{deltas}
and its persistent watershed in \secref{pers-ws}.
In order to be able to extend these results in the network model, we define a
certain class of imprecise terrains which we call \emph{regular} in
\secref{regular}. We prove that persistent watersheds satisfy certain
nesting conditions on regular terrains in \secref{nesting}. This leads to
efficient computations of fuzzy watershed boundaries in \secref{fuzzy-boundary}, and
to the definition of the \emph{fuzzy ridge} in \secref{fuzzy-ridge}, which delineates the persistent
watersheds of the ``main'' minima of a regular terrain and which is equal to the
union of the areas where the potential watersheds of these minima overlap.
We can compute this structure in $O(n \log n)$ time (see
\thmref{compute:fuzzy:ridge}) and we discuss an algorithm that turns a
non-regular terrain into a regular one (see \secref{regularize}).
We conclude the paper with the discussion of open problems in
\secref{conclusions}.

\section{Preliminaries}

We give the definition of imprecise terrains and realizations and discuss the
two flow models used in this paper.

\subsection{Basic definitions and notation}
\label{sec:definitions}

We define an \emphi{imprecise terrain} $T$ as a possibly non-planar geometric
graph $G$ with nodes $V \subset \Re^2$ and edges $E \subseteq V \times V$, where each
node $v \in V$ has an imprecise third coordinate, which represents its
\emphi{elevation}.  We denote the bounds of the elevation of $v$ with \low(v) and
\high(v).  A \emphi{realization} $R$ of an imprecise terrain $T$ consists of the
given graph together with an assignment of elevations to nodes, such that for
each node $v$ its elevation \elev(R,v) is at least \low(v) and at most \high(v).
As such, it is a fully embedded graph $R=(V_R,E_R)$ in $\Re^3$, where $V_R$ is
defined as the set $\{(x,y,z) ~|~ \exists~ v \in V: v=(x,y), z=\elev(R,v)\}$.
Note that this defines a one-to-one correspondence between nodes of $V$ and
$V_R$. The edge set $E_R$ is induced by $E$ under this correspondence.
With slight abuse of notation we will sometimes refer to nodes of $V$ by their
corresponding nodes in $V_R$.
We denote with $\Rlow$ the realization,
such that $\elev(\Rlow,v)=\low(v)$ for every vertex $v$ and similarly the
realization $\Rhigh$, such that $\elev(\Rhigh,v)=\high(v)$.  The set of all
realizations of an imprecise terrain $T$ is denoted with \ReT.


Now, consider a realization $R$ of an imprecise terrain as defined above.
For any set of nodes $P \subseteq V_R$, we define the
\emphi{neighborhood} of $P$ as the set $N(P)=\{s: s \notin P \wedge \exists~ t
\in P : (s,t)\in E_R\}$. If $P$ is a connected set, all nodes of $P$ have the
same elevation and this elevation is strictly lower than the elevation of any
node in  $N(P)$, then $P$ constitutes a \emphi{local minimum}. Likewise, a local
maximum is a set of nodes at the same elevation of which the neighborhood is strictly lower.

\subsection{A model of discrete water flow}
\seclab{def:network:flow}

Consider a realization $R$ of an imprecise terrain as defined above.
If water is only allowed to flow along the edges of the realization, then the
realization represents a network. Therefore we refer to this model of water flow as the
\emphi{network model}. Below, we state more precisely how water flows in this
model and give a proper definition of the watershed. This model or variations of
it have been used before, for example in \cite{Danner2007, o1984extraction,
tarboton1997new}.

The steepness of descent (slope) of an edge $(p,q) \in E_R$ is defined as
$\steep(R,p,q) = (\elev(R,p)-\elev(R,q)) / |pq|$, where $|pq|$ is the Euclidean
distance between the corresponding nodes in $V$.  The node $q$ is a
\emphi{steepest descent neighbor} of $p$, if and only if $\steep(R,p,q)$ is
non-negative and maximal over all neighbours $q$ of $p$. Water that arrives in
$p$ will continue to flow to each of its steepest descent neighbors, unless $p$
constitutes a local minimum.  If there exists a local minimum $P \ni p$, then
the water that arrives in $p$ will flow to the neighbors of $p$ in $P$ and
eventually reach all the nodes of $P$, but it will not flow further to any node
outside the set $P$. If water from $p$ reaches a node $q \in V_R$ then we write $p \flowsto{R} q$
(``$p$ flows to $q$ in $R$''), and for technical reasons we define $p \flowsto{R} p$ for all $p$ and $R$.

The \emphi{discrete watershed} of a node $q$ in a realization $R$ is
defined as the union of nodes that flow to $q$ in $R$, that is
$\WS(R,q):=\{p ~:~ p \flowsto{R}q\}$. Similarly, we define the discrete watershed of
a set of nodes $Q$ in this realization as $\WS(R,Q):=\bigcup_{q\in Q}\WS(R,q)$.

Consider the graph $G$ of the imprecise terrain.
Let $\pi$ be a path in $G$, with no repeated vertices.  We say $\pi$ is a
\emphi{flow path} in a realization $R$ if it carries water to a local minimum
and visits all nodes of this local minimum in $R$.  For two consecutive vertices
$p, q$ in $\pi$, $q$ is a steepest descent neighbour of $p$.
For any pair of nodes $p, q$ in $\pi$, we write $p \flowsvia{\pi} q$ if $p
\flowsto R q$, that is, $\pi$ contains $p$ and $q$ in this order.  We denote
with $\pi[p,q]$ the subpath of $\pi$ from $p$ to $q$, including these two nodes.
For any given set of realizations $S \subseteq \ReT$, we denote with $\FP(S)$
the set of all flow paths in any realization in $S$.

\subsubsection{Flow paths are stable}
\seclab{stable:flow}

This subsection is a note on flow paths, which we defined for the network model
above. We define when a flow path is stable and argue that any flow path induced by
a realization in $\ReT$ is stable with respect to some $\eps-$neighborhood of \ReT.
Intuitively, the analysis in this section shows that the flow paths considered
in our model are never the result of an isolated degenerate situation, but could also
exist if the estimated elevation intervals of the vertices would be slightly different.
This may serve as a justification or proof of soundness of the network model.

If for two realizations $R, R' \in \ReT$ and any node $v\in {V_R}$ and its
corresponding node $v' \in V_{R'}$ we have $|vv'| \leq \eps$, then we call $R'$ an \emphi{$\eps-$perturbation}
of $R$. For a set of realizations $S$, let $S^{\eps}$ denote the union of $S$
with the $\eps-$perturbations of elements of $S$.  We say that a flow
path $\pi \in \FP(S)$ is \emphi{stable} with respect to $S$ if for some $\eps > 0$
the flow path exists in any $\eps-$perturbation of some $R \in S$ (we call $R$
the \emphi{perturbation center}).  Let $\FPstable(S) \subseteq \FP(S)$ denote the
subset of flow paths that is stable with respect to $S$.
We call a realization which does not contain horizontal edges and in which any
node has at most one steepest descent neighbor \emphi{non-ambiguous}, similarly,
a realization for which any of these properties does not hold is called
ambiguous.   We have the following
lemma, which implies that any flow path induced by an element of $\ReT$ is
stable with respect to $\ReT^{\eps}$, for any $\eps>0$.

\begin{lemma}\lemlab{all:stable}
For any set of realizations $S$, we have that $\FP(S) \subseteq \lim_{\eps
\rightarrow 0} \FPstable(S^{\eps})$.
\end{lemma}

\begin{proof}
Given any value $\eps > 0$, we argue that the set  $\FP(S)$ is contained
in the set $\FPstable(S^{\eps})$.  Clearly, any flow path induced by a
non-ambiguous realization $R$ is stable with perturbation center $R$.  Now, let
$\pi=p_1,p_2,\dots,p_k$ be a flow path from $p_1$ to $p_k$ which is induced by
an ambiguous realization $R \in S$.  We lower each node $p_i$ by
$\eps/2+(i\eps)/(4k)$ and perturb the remaining vertices by some value smaller
than $\eps/4$ to create a non-ambiguous realization which also induces $\pi$.
This proves the claim.
\end{proof}


\subsection{A model of continuous water flow}
\seclab{def:surface:flow}

Consider an imprecise terrain, where the graph that represents
the terrain forms a planar triangulation in the $(x,y)$-domain.
Any realization of this terrain is a polyhedral terrain with a triangulated
surface. If we assume that the water which arrives at a particular
point $p$ on this surface will always flow in the true direction of steepest
descent at $p$ across the surface, possibly across the interior of a triangle,
then we obtain a continuous model of water flow. Since the
steepest descent paths do not necessarily follow along the edges of the
graph, but instead lead across the surface formed by the graph,
we call this model the \emphi{surface model}.
This model has been used before, for example in \cite{deberg-10,tsiro-11,
ls-ftt-04}.

Since, as we will show in the next section, it is already NP-hard to decide whether water from a point $p$ can potentially flow to another point $q$, we will focus on the network model in the rest of the paper, and we do not formally define watersheds in the surface model. Therefore, we will simply use the term ``watershed'' to refer to discrete watersheds in this paper.




\section{NP-hardness in the surface model}
\seclab{surface}

In the surface model water flows across the surface of a polyhedral
terrain; refer to \secref{def:surface:flow} for the details of the model.
In this section we prove that it is NP-hard to decide whether water potentially
flows from a point $s$ to another point~$t$ in this model. The reduction is from
\mbox{3-SAT}; the input is a 3-CNF formula with $n$ variables and $m$ clauses.
We first present the general idea of the proof, then we proceed with a
detailled description of the construction, and finally we prove the correctness.


\subsection{Overview of the construction}
The main idea of the NP-hardness construction is to encode the variables and
clauses of the \mbox{3-SAT} instance in an imprecise terrain, such that a truth
assignment to the variables corresponds to a realization---i.e., an assignment
of elevations---of this terrain.  If and only if all clauses are satisfied,
water will flow from a certain starting vertex $s$ to a certain target vertex $t$.
We first introduce the basic elements of the construction: channels and switch gadgets.

\begin{wrapfigure}[6]{r}{0.3\textwidth}  \centering
\includegraphics[scale=0.6]{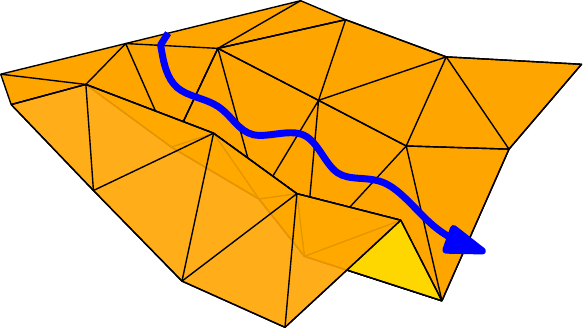}
\end{wrapfigure}
\paragraph{Channels}
We can mold channels in the fixed part of the terrain to route water
along any path, as long as the path is monotone in the direction of steepest
descent on the terrain. We do this by increasing the elevations of vertices
next to the path, thus building walls that force the water to stay in the
channel. We can end a channel in a local minimum anywhere
on the fixed part of the terrain, if needed.

\paragraph{Switch gadgets}
The general idea of a \emph{switch gadget} is that it provides a way for water
to switch between channels.  A simple switch gadget has one incoming channel,
three outgoing channels, and two \emph{control vertices} $a$ and $b$, placed on
the boundary of the switch. The water from the incoming channel has to flow
across a central triangle, which is connected to $a$ and~$b$.
Depending on their elevations, the two vertices $a$ and $b$ divert the water
from the incoming channel to a particular outgoing channel and thereby
``control'' the behaviour of the switch gadget.  This is possible, since the
slope of the central triangle, which the water needs to pass, depends on the
elevations of $a$ and $b$ and those two are the only vertices with imprecise
elevations. The elevations of the vertices which are part of the channels are
fixed.  Refer to \figref{hardness-switch} for an illustration.

\begin{figure}[thb]\center
\includegraphics[width=0.3\textwidth,page=2]{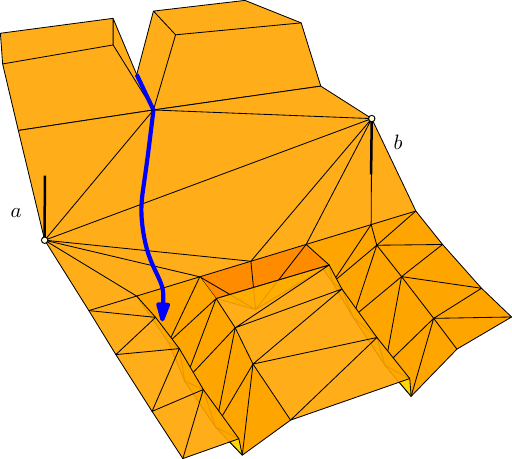}
\includegraphics[width=0.3\textwidth,page=4]{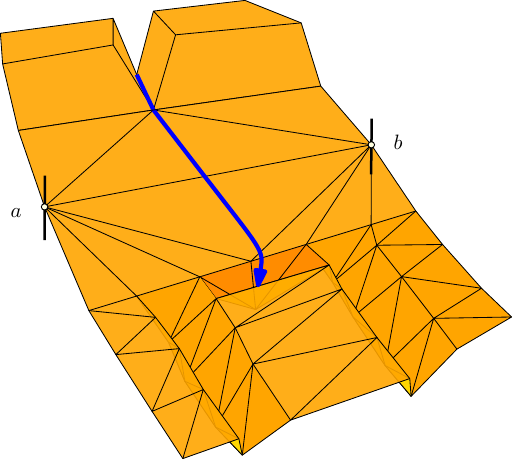}
\includegraphics[width=0.3\textwidth,page=3]{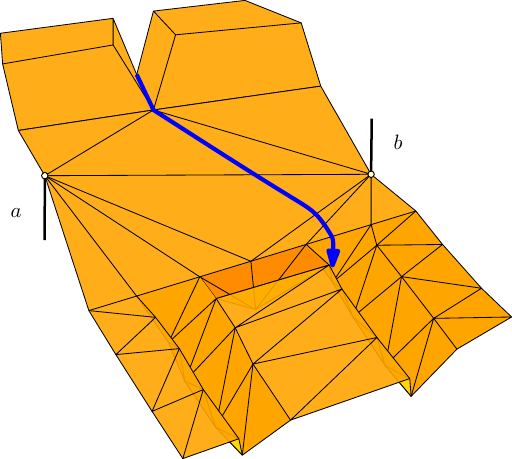}
\caption{Three different states of a simplistic switch gadget. }
\figlab{hardness-switch}
\end{figure}

We can also build switches for multiple incoming channels. In this case, every
incoming channel has its own dedicated set of outgoing channels, and it is also
controlled by only two vertices, see \figref{hardness-switch-multiple}. Note
that we can lead the middle outgoing channel to a local minimum as shown in the
examples and, in this way, ensure that, if any water can pass the switch, the
elevations of its control vertices are at unambigous extremal elevations.
Depending on the exact construction of the switch, we may want them to be at
opposite extremal elevations or at corresponding extremal elevations.

\begin{figure}[thb]\center
\includegraphics[width=0.45\textwidth,page=2]{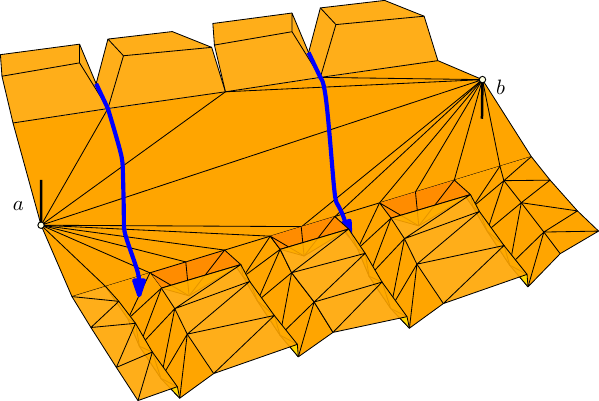}
\includegraphics[width=0.45\textwidth,page=3]{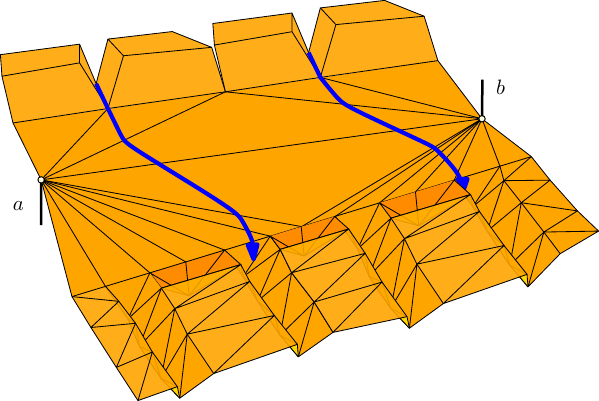}
\caption{Sketch of a switch with multiple incoming channels.}
\figlab{hardness-switch-multiple}
\end{figure}

\paragraph{Global layout}
The global layout of the construction is depicted in
Figure~\ref{fig:hardness-global}. The construction contains a grid of $m \times
n$ cells, in which each clause corresponds to a column and each variable to a
row of the grid. The grid is placed on the western slope of a ``mountain'';
columns are oriented north-south and rows are oriented east-west. We create a
system of channels that spirals around this mountain, starting from $s$ at the
top and ending in $t$ at the bottom of the mountain.
We ensure that in no realization, water from $s$ can escape this channel system
and, if it reaches $t$, we know that it followed a strict course that passes
through every cell of the grid exactly once, column by column from east to
west, and within each column, from north to south.
Embedded in this channel system, we place a switch gadget in every cell
of the grid, which allows the water from $s$ to ``switch'' from one channel to another
within the current column depending on the elevations of the vertices that control the
gadget.  In this way, the switch gadgets of a row encode the state of a variable.
To ensure that the state of a variable is encoded consistently across a row of
the grid, the switch gadgets in a row are linked by their control vertices.
Every column has a dedicated entry point at its north end, and a dedicated exit
point at its south end. If and only if water flows between these two points,
the clause that is encoded in this column is satisfied by the corresponding
truth assignment to the variables.
The slope of the mountain is such that columns descend towards the south,
and the exit point of each column (except the westernmost one) is lower than
the entry point of the adjacent column to the west; water can flow between
these points through a channel around the back of the mountain. The easternmost
column's entry point is the starting vertex $s$, and the westernmost column's
exit point is the target vertex $t$.

%
%

\begin{figure}[t]
\centering
\includegraphics[width=0.60\textwidth]{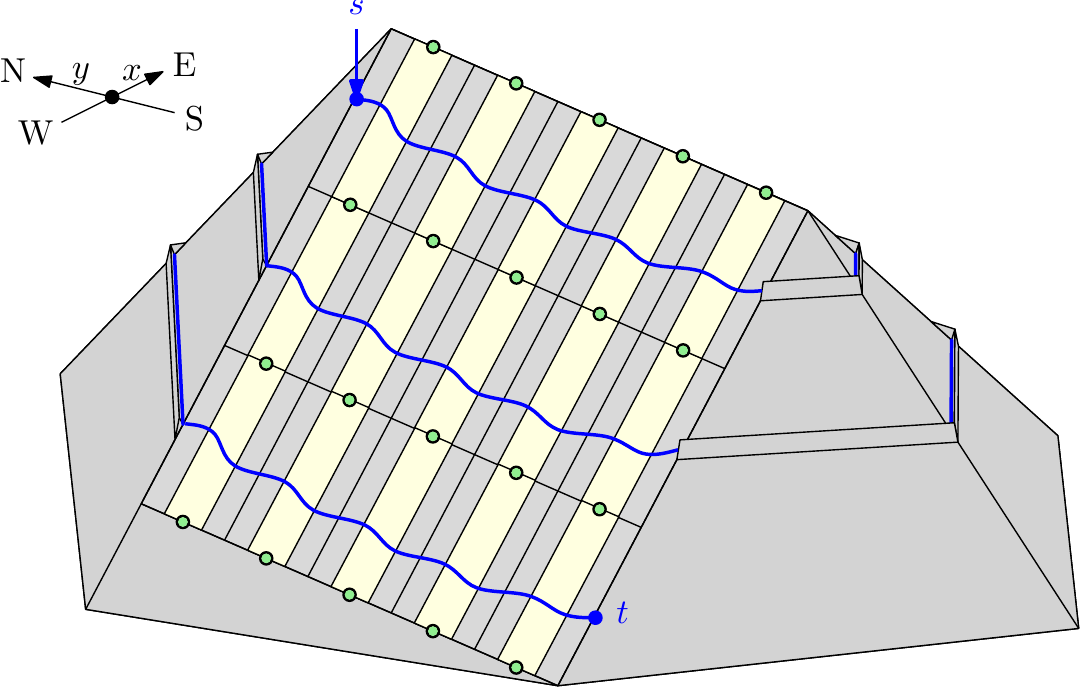}
\hfil
\includegraphics[width=0.25\textwidth]{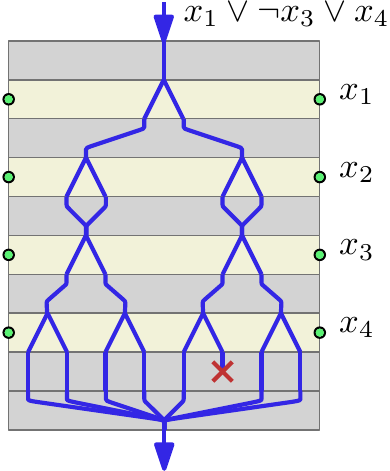}
\caption {
Left: Global view of the NP-hardness construction, showing the grid on the
mountain slope.  The fixed parts are shown in gray, the variable parts are
shown light yellow and the imprecise vertices are filled light green;
Right: Detail of a clause, which forms one of the columns of the grid.}
\label{fig:hardness-global}
\end{figure}

%
\paragraph{Clause columns}
To encode each clause, we connect the switch gadgets in a column of the global grid by
channels in a tree-like manner. By construction, water will arrive in a different channel
at the bottom of the column for each of the eight possible combinations of truth
values for the variables in the clause. This is possible
because a switch gadget can switch multiple channels simultaneously. We let the channel
in which water would end up if the clause is not satisfied lead to a local
minimum; the other seven channels merge into one channel that leads to the exit
point of the clause. The possible courses that water can take will also cross switch
gadgets of variables that are not part of the clause: in that case, each course
splits into two courses, which are merged again immediately after emerging from
the switch gadget. Figure~\ref {fig:hardness-global} (right) shows an example.


\paragraph{Sloped switch gadgets}
Since the grid is placed on the western slope of a mountain, water on the
central triangle of a switch will veer off towards the west, regardless of the
elevations of its control vertices.  However, as we will see, we can still
design a working switch gadget in this case.  Recall that we link the switch
gadgets of a variable row by their control vertices, such that each switch
gadget shares one control vertex with its neighboring cell to the west and one
with its neighboring cell to the east.  As mentioned before, such a row encodes the
state of a particular variable. We say that it is in a consistent state if
either all control vertices of the switches are \emph{high}
or all control vertices are \emph{low}.  Thus, we will use the following
assignment of truth values to the elevations of the control vertices of our
switch gadgets: both vertices set to their highest elevation
encodes \emph{true}; both vertices set to their lowest elevation
encodes \emph{false}; other combinations encode \emph{confused}. Depending
on the truth value encoded by the elevations of the imprecise vertices, water
that enters the gadget will flow to different channels. The channels in which
the water ends up when the gadget reads \emph{confused} always lead to a local
minimum. For the other channels, their destination depends on the clause.
In \figref{hardness-switch-sloped} you can see a sketch of a sloped switch
gadget which works the way described above.

\begin{figure}[htb]\center
\includegraphics[width=0.3\textwidth,page=1]{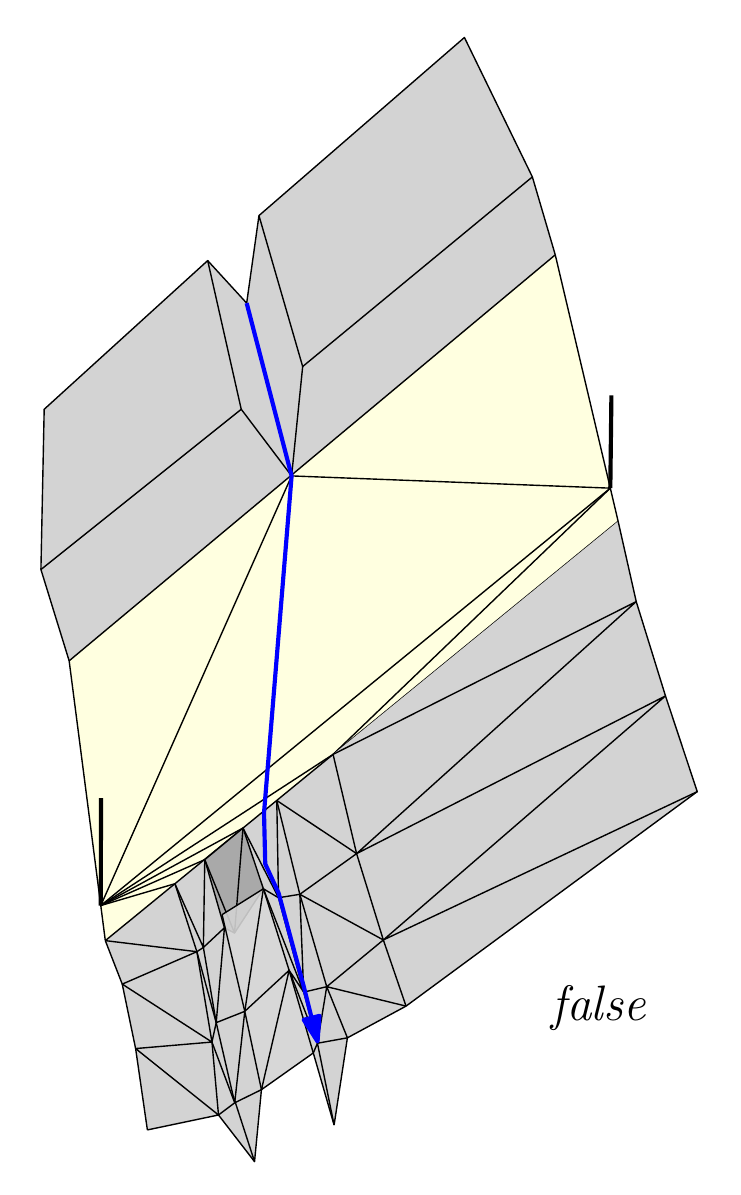}
\includegraphics[width=0.3\textwidth,page=3]{figs/sloped_switch_color}
\includegraphics[width=0.3\textwidth,page=2]{figs/sloped_switch_color}
\caption{Illustration of a sloped switch gadget similar to the one used in the final construction.
The final gadget has multiple incoming channels, which is not shown in this
figure. }
\figlab{hardness-switch-sloped}
\end{figure}

\subsection{Details of the construction}

Recall that we are given a 3-SAT instance with $n$ variables and $m$ clauses. The central part of the construction, which will contain the gadgets, consists of a grid of $n$ rows---one for each variable---and $m$ columns---one for each clause. We denote the width of each row, measured from north to south, by $B = 400$, and the width of a column, measured from west to east, by $A = \max((n+1) \cdot B, 4000)$. Ignoring local details, on any line from north to south in this part of the construction, the terrain descends at a rate of $dz/dy = 1$, and on any line from east to west, it descends at a rate of $dz/dx = 1$; thus we have $z = x+y$. Observe that each column measures $n B < A$ from north to south; thus the southern edge of each column is at a higher elevation than the northern edge of the next column to the west. The dedicated entrance and exit points of column $1 \leq j \leq n$ are placed at
$(jA - \frac12A, nB, jA - \frac12A + nB)$ and
$(jA - \frac12A, 0, jA - \frac12A)$, thus allowing the construction of a descending channel from each column's exit point to the entry point of the column to the west.

For every variable $v_i$, $1 \leq i \leq n$, we place $m+1$ imprecise vertices $v_{ij}$, for $0 \leq j \leq n$,\break in row~$i$, on the boundaries of the columns corresponding to the $m$ clauses. Vertex $v_{ij}$ has $x$-coordinate~$jA$, $y$-coordinate $iB - \frac12B$, and an imprecise $z$-coordinate $[jA+iB-\frac12B,\break jA+iB-\frac12B+20]$. On every pair of imprecise vertices $v_{i(j-1)}, v_{ij}$ we build a switch gadget $G_{ij}$; thus there is a switch gadget for each variable/clause pair. The coordinates of the vertices in each gadget, relative to the coordinates of $v_{i(j-1)}$, can be found in \figref{gradients}.



\paragraph {Switch gadget construction}
We use the sloped switch gadget described above and illustrated in \figref{hardness-switch-sloped}.
Our switch gadget occupies a rectangular area that is $A$ wide from west to east,
and $41$ wide from north to south. Its key vertices and their coordinates,
relative to each other, can be found in \figref{gradients}. There are two
imprecise vertices, $d$ and $e$, with elevation range $[0,20]$ and $[A, A+20]$,
respectively---so in any realization, their elevations have the form $20\alpha$
and $A+20\beta$, respectively, where $\alpha, \beta \in [0,1]$.

\begin{figure}
\centering
\includegraphics[scale=1.15]{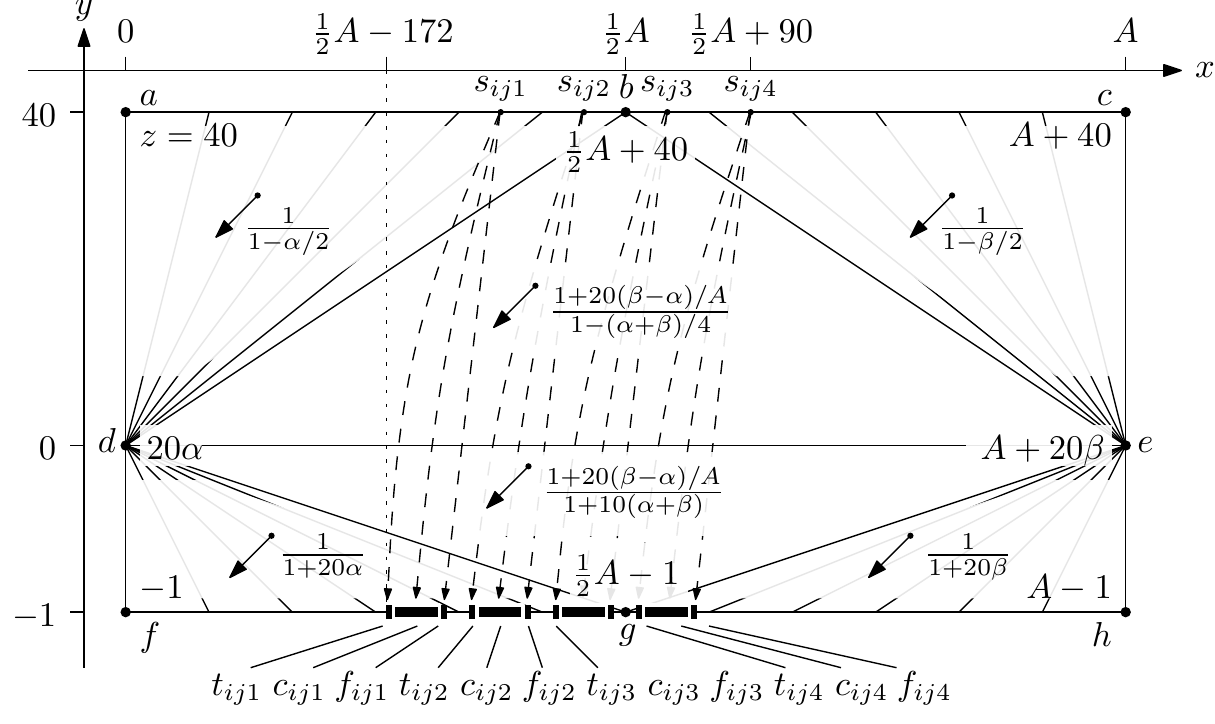}
\caption{Distances and gradients on a connector gadget. All coordinates are relative to the lowermost position of the control vertex $d = v_{i(j-1)}$. The other control vertex is $e = v_{ij}$. Thus, $e$ and $d$ are the only imprecise vertices. The $x$- and $y$-coordinates of the vertices are indicated on the axes. The elevations of the key vertices are written next to the vertices. The elevations of the control vertices are expressed as a function of $\alpha, \beta \in [0,1]$. The directions of steepest descent on the different faces of the gadget are expressed in the form $dy/dx$, as a function of $\alpha$ and $\beta$.}
\figlab{gradients}
\end{figure}

On the north edge of the gadget, there may be many more vertices, all collinear with $a$, $b$ and $c$. The vertices on the western half of the north edge are connected to the western control vertex, and the vertices on the eastern half of the north edge are connected to the eastern control vertex. In particular, each gadget $G_{ij}$ is designed to receive water from four channels that arrive at four points $s_{ij1}, s_{ij2}, s_{ij3}, s_{ij4}$ on the north edge, close to $b$; the coordinates of these points are $s_{ijk} = (\frac12A - 150 + 60k, 40, \frac12A - 110 + 60k)$.

On the south edge of the gadget, there is a similar row of vertices, all
collinear with $f$, $g$ and $h$, that are connected to the control vertices. To
the south, the gadget is connected to twelve channels that catch all water that
arrives at certain intervals on the south edge: for each $k \in \{1,2,3,4\}$,
there is a \emph{western channel} $t_{ijk}$ catching all water arriving between
$s_{ijk} - (82,41,123)$ and $s_{ijk} - (77,41,118)$, a \emph{middle channel}
$c_{ijk}$ catching all water arriving between $s_{ijk} - (77,41,118)$ and
$s_{ijk} - (44,41,85)$, and an \emph{eastern channel} $f_{ijk}$ catching all
water arriving between $s_{ijk} - (44,41,85)$ and $s_{ijk} - (39,41,80)$.

In a particular realization $R$, we define the switch gadget to be in a
\emph{false} state if $\alpha = \beta = 0$, in a \emph{true} state if $\alpha =
\beta = 1$, and in a \emph{confused} state if $\alpha \leq \frac12$ while $\beta
\geq \frac12$, or if $\alpha \geq \frac12$ while $\beta \leq \frac12$. As we
will show below, in the true, false, and confused states the gadget leads any
water that comes in at any point $s_{ijk}$ into $t_{ijk}$, $f_{ijk}$, and
$c_{ijk}$, respectively.


We model the fixed part of the terrain such that the middle channels all lead to
local minima. The western and eastern channels correspond to a (partial) truth
assignment of the variables of the clause that is represented by the column that
contains the gadget; these channels lead to a local minimum or to the next row,
as described below.

\paragraph{Constructing the clause columns}

Each clause is modeled in a column $j$ by making certain connections between the outgoing channels of each gadget to the dedicated entrance points of the gadget in the next row. Observe that by our choice of $B$, the entrance point of column $j$ lies above all entrance points of $G_{nj}$, all outgoing channels of any gadget $G_{(i+1)j}$ start at higher elevations than all entrance points of $G_{ij}$, and all outgoing channels of $G_{1j}$ start at an elevation higher than the exit point of the column. This ensures that all channels described below can indeed be built as monotonously descending channels, so that water can flow through it. We will now explain the connections which we use to build a clause.

Let $p > q > r$ be the indices of the variables that appear in the clause. The water courses modelling the clause start at the entry point of the column, from which any water is led through a channel to entry point $s_{nj1}$ of gadget $G_{nj}$.

For $i \neq \{p,q,r\}, k \in \{1,2,3,4\}$, we connect both $t_{ijk}$ and $f_{ijk}$ to $s_{(i-1)jk}$ (if $i > 1$) or to the exit point of the column (if $i = 1$).

We connect $t_{pj1}$ and $f_{pj1}$ to $s_{(p-1)j1}$ and $s_{(p-1)j2}$, respectively. Thus, for $i \in \{q,...,p-1\}$, water that enters $G_{ij}$ at $s_{ij1}$ and $s_{ij2}$ represents the cases that $p$ is true and $p$ is false, respectively.

We connect $t_{qj1}, f_{qj1}, t_{qj2}$ and $f_{qj2}$ to $s_{(q-1)j1}, s_{(q-1)j2}, s_{(q-1)j3}$ and $s_{(q-1)j4}$, respectively. Thus, for $i \in \{r,...,q-1\}$, water that enters $G_{ij}$ at $s_{ij1}, s_{ij2}, s_{ij3}$ and $s_{ij4}$ represents the four different possible combinations of truth assignments to $p$ and $q$, respectively.

The eight channels $t_{rj1}, f_{rj1}, t_{rj2}, f_{rj2}, t_{rj3}, f_{rj3}, t_{rj4}, f_{rj4}$ now represent the eight different possible combinations of truth assignments to the variables of the clause. The channel that corresponds to the truth assignment that renders the clause false, is constructed such that it ends in a local minimum. The other seven channels all lead to $s_{(r-1)j1}$ (if $r > 1$) or to the exit point of the clause column (if $r = 1$).

\paragraph{Analysis of flow through a gadget}
Below we will analyse where water may leave a gadget $G_{ij}$ after entering the gadget at point $s_{ijk}$, with $x$-coordinate $x_k$. In the discussion below, all coordinates are relative to the lowermost position of the western control vertex of the gadget---refer to \figref{gradients}, which also shows the directions of steepest descent (expressed as $dx/dy$) on each face of the gadget.

First observe that in any case, the directions of steepest descent on $\triangle abd$, $\triangle bce$ and $\triangle bed$ are at least $1-20/A \geq 199/200 = 0.995$ and at most $(1+20/A)/(1/2) \leq 201/100 = 2.01$. Thus, when the water reaches $y$-coordinate $38$, it will be at $x$-coordinate at least $x_k - 4.02$ and at most $x_k - 1.99$.

Note that the line $bd$ intersects the plane $y = 38$ at $x = \frac12 A - \frac1{40} A \leq \frac12 A - 100$, and the line $be$ intersects the plane $y = 38$ at $x = \frac12 A + \frac1{40} A \geq \frac12 A + 100$. By our choice of coordinates for the entrance points $s_{ijk}$, we have $|x_k - \frac12A| \leq 90$; therefore the water will be on $\triangle bed$ when it reaches $y = 38$. Let $g_{\max}$ and $g_{\min}$ be the maximum and minimum possible gradients $dx/dy$ on $\triangle bed$, respectively. Thus, the water will reach the line $de$ at $x$-coordinate at least $x_k- 4.02 - 38 g_{\max}$ and at most $x_k - 1.99 - 38 g_{\min}$.

Finally, the directions of steepest descent on $\triangle dgf$, $\triangle egh$ and $\triangle deg$ are more than 0 and less than $1 + 20/A \leq 201/200 < 1.01$. Thus, the water will reach the line $fh$ at $x$-coordinate more than $x_k - 5.03 - 38 g_{\max}$ and less than $x_k - 1.99 - 38 g_{\min}$.

We will now consider five classes of configurations of the control vertices in the gadget, and compute the interval of $x$-coordinates where water may reach the line $fh$ in each case.

\begin{itemize}
\item \textit{$\alpha = \beta = 1$ (true state)}
In this case we have $g_{\max} = g_{\min} = 2$, so water will reach the line $fh$ within the $x$-coordinate interval $(x_k - 81.03, x_k - 77.99)$, and thus it will flow into channel $t_{ijk}$.

\item \textit{$\alpha + \beta > \frac32$ (true-ish state)}
In this case we have $g_{\max} \leq (1+20/A)/(1/2) \leq 2.01$ and $g_{\min} \geq (1-20/A)/(1-3/8) \geq 199/125 > 1.59$. Thus water will reach the line $fh$ within the $x$-coordinate interval $(x_k - 81.41, x_k - 62.41)$, and thus it will flow into channel $t_{ijk}$ or $c_{ijk}$.

\item \textit{$\frac12 \leq \alpha+\beta \leq \frac32$ (this includes all confused states)}
In this case we have $g_{\max} \geq (1+20/A)/(1-3/8) \leq 201/125 < 1.61$ and $g_{\min} \geq (1-20/A)/(1-1/8) \geq 199/175 > 1.13$. Thus water will reach the line $fh$ within the $x$-coordinate interval $(x_k - 66.21, x_k - 44.93)$, and thus it will flow into channel $c_{ijk}$.

\item \textit{$\alpha + \beta < \frac12$ (false-ish state)}
In this case we have $g_{\max} \leq (1+20/A)/(1-1/8) \leq 201/175 < 1.15$ and $g_{\min} \geq (1-20/A) \geq 199/200 > 0.99$. Thus water will reach the line $fh$ within the $x$-coordinate interval $(x_k - 48.73, x_k - 39.61)$, and thus it will flow into channel $c_{ijk}$ or $f_{ijk}$.

\item \textit{$\alpha = \beta = 0$ (false state)}
In this case we have $g_{\max} = g_{\min} = 1$, so water will reach the line $fh$ within the $x$-coordinate interval $(x_k - 43.03, x_k - 39.99)$, and thus it will flow into channel $f_{ijk}$.
\end{itemize}

\paragraph{Correctness of the NP-hardness reduction}
\begin{lemma}
If water flows from $s$ to $t$ in some realization, then there is a truth assignment of the variables of the 3-CNF formula that satisfies the formula.
\end{lemma}

\begin{proof}
Water that starts flowing from $s$, which is the entrance point of the clause
column $m$, is immediately forced into a channel to entrance point $s_{nm1}$ of
gadget $G_{nm}$. As calculated above, any water that enters a gadget at one of
its designated entrance points will leave the gadget in one of its designated
channels, which leads either to a local minimum, or to a designated entrance
point of the next gadget. Therefore, water from $s$ can only reach $t$ after
flowing through all switch gadgets.

Since all middle outgoing channels $c_{ijk}$ lead to local minima, we know that if there is a flow path from $s$ to $t$, then the water from $s$ is nowhere forced into a middle outgoing channel. It follows that no gadget is in a confused state. As a consequence, in any row, either all gadgets have their control vertices in the lower relatively open half of their elevation range, or all gadgets have their control vertices in the upper relatively open half of their elevation range. In the first case, all gadgets in the row are in a \emph{false-ish} state, and any incoming water from $s$ leaves those gadgets in the same channels as if the gadgets were in a proper \emph{false} state. In the second case, all gadgets in the row are in a \emph{true-ish} state, and any incoming water from $s$ leaves those gadgets in the same channels as if the gadgets were in a proper \emph{true} state.

We can now construct a truth assignment ${\cal A}$ to the variables, in which each variable is \emph{true} if the control vertices in the corresponding row are in the upper halves of their elevation ranges, and \emph{false} otherwise. It follows from the way in which channel networks in clause columns are constructed, that in each clause column, water will flow into one of the seven channels that corresponds to a truth assignment that satisfies the corresponding clause---otherwise the water would not reach $t$. Therefore, ${\cal A}$ satisfies each clause, and thus, the complete 3-CNF formula.
\end{proof}

\begin{lemma}
If there is a truth assignment to the variables that satisfies the given 3-CNF formula, then there is a realization of the imprecise terrain in which water flows from $s$ to $t$.
\end{lemma}
\begin{proof}
We set all control vertices in rows corresponding to true variables to their highest positions and all control vertices in rows corresponding to false variables to their lowest positions. One may now verify that, by construction, in each clause column water from the column's entry point will flow into one of the seven channels that lead to the column's exit point, and thus, water from $s$ reaches $t$.
\end{proof}


\vspace{\baselineskip}
Thus, 3-SAT can be reduced, in polynomial time, to deciding whether there is a realization of $T$ such that water can flow from $s$ to $t$. We conclude that deciding whether there exists a realization of $T$ such that water can flow from $s$ to $t$ is NP-hard.

\begin {theorem}
  Let $T$ be an imprecise triangulated terrain, and let $s$ and $t$ be two
  points on the terrain. Deciding whether there exists a realization $R \in
  \ReT$ such that $p \flowsto{R} q$ is NP-hard.
\end {theorem}

\section{Watersheds in the network model}
\seclab{network}

In the network model we assume that water flows only along the edges of a
realization.  More specifically, water that arrives in a node $p$ continues to
flow along the steepest descent edges incident on $p$, unless $p$ is a local
minimum.  For a formal definition of the watershed and flow paths
please refer to \secref{def:network:flow}.

\subsection{Potential watersheds}
\seclab{potential-ws}
The \emphi{potential watershed} of a set of nodes $Q$ in a terrain $T$ is defined as
\[\PoWS(Q):=\bigcup_{R \in \ReT} \bigcup_{q \in Q} \WS(R,q),\]
which is the union of the watersheds of $Q$ over all realizations of $T$.
This is the set of nodes for which there exists a flow path to a node of
$Q$. With slight abuse of notation, we may also write $\PoWS(q)$ to denote the
potential watershed of a single node $q$.


\subsubsection{Canonical realizations}

We prove that for any given set of nodes $Q$ in an imprecise terrain, there exists a
realization $R$ such that $\WS(R,Q) = \PoWS(Q)$.
For this we introduce the notion of the overlay of a set of watersheds in
different realizations of the terrain. Informally, the overlay is a realization
that sets every node that is contained in one of these watersheds to the lowest
elevation it has in any of these watersheds.

\begin{defn}\label{def:ws-overlay}
Given a sequence of realizations $R_1,...,R_k$ and a sequence of nodes
$q_1,...,q_k$, the \emphi{watershed-overlay} of $\WS(R_1,q_1),...,\WS(R_k,q_k)$
is the realization \WSO such that for every node $v$, we have that
$\elev(\WSO,v)=\high(v)$ if $v \notin \bigcup \WS(R_i,q_i)$ and otherwise
\[\elev(\WSO,v) = \min_{i:v \in \WS(R_i,q_i)} \elev(R_i,v).\]
\end{defn}


\begin{lemma}\label{lem:watershed-overlay}
Let $\WSO$ be the watershed-overlay of $\WS(R_1,q_1),\dots,\WS(R_k,q_k)$, and
let $Q=\bigcup_{1\leq i\leq k} q_i$, then $\WS(\WSO,Q)$ contains $\WS(R_i,q_i)$.
\end{lemma}
\begin{proof}
Let $u$ be a node of the terrain. We prove the lemma by induction on increasing symbolic
elevation to show that if $u$ is contained in one of the given watersheds, then
it is also contained in $\WS(\WSO,Q)$.
We define $\level(R_i,u)$ as the smallest number of edges on any path along
which water flows from $u$ to $q_i$ in $R_i$; if there is no such path, then
$\level(R_i,u) = \infty$. Now we define the \emphi{symbolic elevation} of $u$,
denoted $\symb(u)$, as follows: if $u$ is contained in any watershed
$\WS(R_i,q_i)$, then $\symb(u)$ is the lexicographically smallest tuple
$(\elev(R_i,u), \level(R_i,u))$ over all $i$ such that $u \in \WS(R_i,q_i)$;
otherwise $\symb(u) = (\high(u), \infty)$.

Now consider a node $u$ that is contained in one of the given watersheds.
The base case is that $u$ is contained in $Q$, and in this case the claim holds trivially.
Otherwise, let $R_i$ be a realization such that $u \in \WS(R_i,q_i)$ and such
that $(\elev(R_i,u), \level(R_i,u))$ is lexicographically smallest over all $1
\leq i \leq k$. By construction, we have that $\elev(R_i,u) = \elev(\WSO,u)$.
Consider a neighbour $v$ of $u$ such that $(u,v)$ is a steepest-descent edge
incident on $u$ in $R_i$, and $\level(R_i,v)$ is minimal among all such
neighbours $v$ of $u$. Since $\elev(\WSO,v) \leq \elev(R_i,v) \leq \elev(R_i,u)
= \elev(\WSO,u)$ and $\level(R_i,v) = \level(R_i,u) - 1$, it holds that $v$ has
smaller symbolic elevation than $u$. Therefore, by induction, $v \in
\WS(\WSO,Q)$. If $v$ is still a steepest descent neighbor of $u$ in \WSO, then
this implies $u \in \WS(\WSO,Q)$.  Otherwise, there is a node $\widehat{v}$
such that $\steep(\WSO,u,\widehat{v}) > \steep(\WSO,u,v) \geq 0$. There must be
an $R_j$ such that $\widehat{v} \in \WS(R_j,q_j)$, since otherwise, by
construction of the watershed-overlay, we have $\elev(\WSO,\widehat{v}) =
\high(\widehat{v}) \geq \elev(R_i,\widehat{v})$ and thus,
$\steep(R_i,u,\widehat{v}) \geq \steep(\WSO,u,\widehat{v}) > \steep(\WSO,u,v)
\geq \steep(R_i,u,v)$ and $v$ would not be a steepest descent neighbor of $u$ in
$R_i$. Moreover, we have $\steep(\WSO,u,\widehat{v}) > 0$ and, therefore,
$\elev(\WSO,\widehat{v}) < \elev(\WSO,u)$, so $\widehat{v}$ has smaller symbolic
elevation than $u$. Therefore, by induction, also $\widehat{v} \in \WS(\WSO,Q)$
and thus, $u \in \WS(\WSO,Q)$.
\end{proof}

\bigskip
The above lemma implies that for any set of nodes $Q$, the watershed-overlay
$\WSO$ of the watersheds of the elements of $Q$ in all possible
realizations $\ReT$, would realize the potential watershed of $Q$. That is, we have that
$\PoWS(Q) \subseteq \WS(\WSO,Q)$ and since $\PoWS(Q)$ is the union of all
watersheds of $Q$ in all realizations, we also have that $\WS(\WSO,Q) \subseteq
\PoWS(Q)$, which implies the equality of the two sets.
Therefore, we call $\WSO$ the \emphi{canonical realization} of the potential
watershed $\PoWS(Q)$ and we denote it with
$\CanonR(Q)$.

Note, however, that it is not immediately clear that the canonical realization always exists: the set of possible realizations is a non-discrete set, and thus the elevations in the canonical realization are defined as minima over a non-discrete set. Therefore, one may wonder if these minima always exist. Below, we will describe an algorithm that can actually compute the canonical realization of any set of nodes $Q$; from this we may conclude that it always exists.

\subsubsection{Outline of the potential watershed algorithm}

Next, we describe how to compute $\PoWS(Q)$ and its canonical realization
$\CanonR(Q)$ for a given set of nodes $Q$. Note that for all nodes $p \notin
\PoWS(Q)$, we have, by definition of the canonical realization,
$\elev(\CanonR(Q),p) = \high(p)$. The challenge is therefore to compute
$\PoWS(Q)$ and the elevations of the nodes of $\PoWS(Q)$. Below we describe an
algorithm that does this.

The idea of the algorithm is to compute the nodes of $\PoWS(Q)$ and their
elevations in the canonical realization in increasing order of elevation,
similar to the way in which \hbox{Dijkstra}'s shortest path algorithm computes
distances from the source. The complete algorithm is laid out in
\algref{compute-pws}.
The correctness and running time of the algorithm are
proved in \thmref{compute-pws}. A key ingredient of the algorithm is a
subroutine, $\expandPWSX{q'}{z'}$, which is defined as follows.

\begin{defn}
Let $\expandPWSX{q'}{z'}$ denote a function that returns for a node $q'$
and an elevation $z' \in [\low(q'),\high(q')]$ a set of pairs of nodes and
elevations, which includes the pair $(p,z)$ if and only if $p \in N({q'})$,
there is a realization $R$ with $\elev(R,q') \in [z',\high(q')]$ such that $p
\flowsto{R} q'$, and $z$ is the minimum elevation of $p$ over all such
realizations $R$.
\end{defn}

\begin{algorithm}
  \begin{algorithmic}[1]
    \STATE For all $q \in Q$: Enqueue $(q,z)$ with key $z=\low(q)$
    \WHILE{the Queue is not empty}
          \STATE     $(q',z')$ = DequeueMin()
          \IF{$q'$ is not already in the output set}
               \STATE  Output $q'$ with elevation $z'$
               \STATE  Enqueue each $(p,z) \in \expandPWSX{q'}{z'}$
          \ENDIF
    \ENDWHILE
  \end{algorithmic}
\caption{\computePWSX{Q}}
\label{compute-pws}
\alglab{compute-pws}
\end{algorithm}

\subsubsection{Expansion of a node using the slope diagram}
\seclab{slope-diagram}
\seclab{expand-pws}

Before presenting the algorithm for the expansion of a node, we discuss a
data structure that allows us to do this efficiently.

We define the \emphi{slope diagram} of a node $p$ as the set of points
$\widehat{q}_i = (\delta_i, \high(q_i))$, such that $q_i$ is a neighbor of $p$ and
$\delta_i$ is its distance to $p$ in the $(x,y)$-projection.
Let $q_1,q_2,...,$ be a subset of the neighbors of $p$ indexed such that
$\widehat{q}_1, \widehat{q}_2, ...$ appear in counter-clockwise order along the
boundary of the convex hull of the slope diagram, starting from the leftmost
point and continuing to the lowest point. We ignore neighbors that do not lie on
this lower left chain.

Let $H_i$ be the halfplane in the slope diagram that lies above the
line through $\widehat{q}_i$ and $\widehat{q}_{i+1}$. Let $U(p)$ be the
intersection of these halfplanes $H_1,H_2,...$, the halfplane right of the
vertical line through the leftmost point, and the halfplane above the
horizontal line through the bottommost point of the convex chain, see
the shaded area in \figref{slopediagram}.
Let $z_i$ be the $z$-value of the point where the line through
$\widehat{q}_i$ and $\widehat{q}_{i+1}$ intersects the vertical axis
of the slope diagram.

\begin{figure}[b]\center
\includegraphics[width=0.4\textwidth]{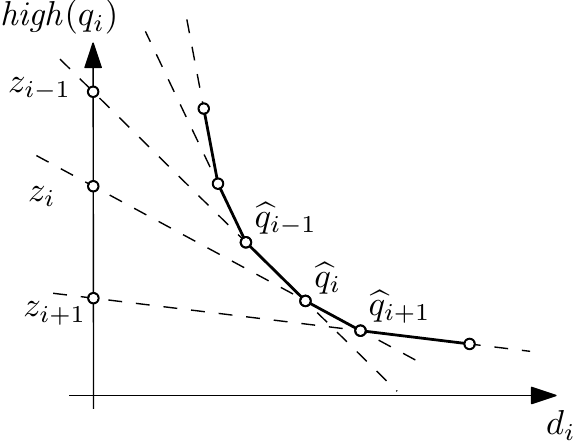}
\hfil
\includegraphics[width=0.4\textwidth]{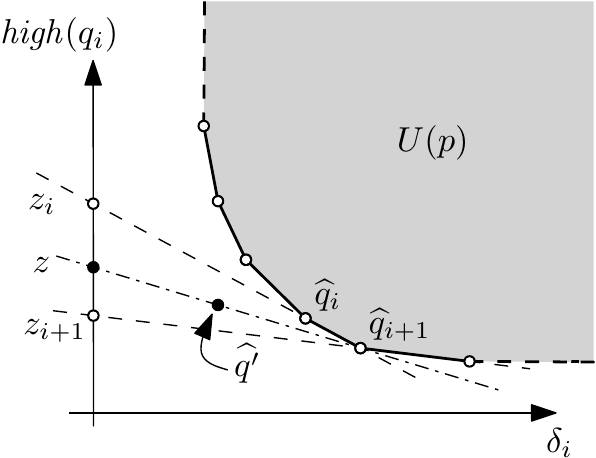}
\caption{Left: slope diagram. Right: querying the slope diagram.}
\figlab{slopediagram}
\end{figure}

Note that, out of the neighbors of $p$, each set to their highest position,
$q_i$ is a steepest descent neighbor if and only if the elevation of $p$ lies in
the interval $(z_i, z_{i-1})$.  This observation will help us to compute the
minimum elevation of~$p$ such that water flows to any particular neighbor in
$O(\log d)$ time, given that we have the slope diagram of $p$ at hand.


For a neighbor $p$ of $q'$, we can now compute the elevation of $p$ as it should be
returned by \expandPWSX{q'}{z'} by computing the lower tangent to $U(p)$ which passes
through the point $\widehat{q'}=(\delta', z')$, where $\delta'$ is the distance from $q'$ to $p$ in
the $(x,y)$-projection. This can be done via a binary search on the boundary of
$U(p)$. Intuitively, this tangent intersects the corner of $U(p)$ which
corresponds to the neighbor of $p$ that the node $q'$ has to compete with
for being the steepest-descent neighbor of $p$.
The elevation $z$ at which the tangent intersects the vertical axis, is the
lowest elevation of $p$ such that $q'$ does not lose, see the figure.
In proving the following lemma we describe the details of this procedure more specifically.

\begin{lemma}\lemlab{pre-expand-pws}
Given the slope diagrams of the neighbours of $q'$, we can compute the function \expandPWSX{q'}{z'} in time $O(d \log d')$,
where $d$ is the node degree of $q'$, $d'$ is the maximum node degree of a neighbor of $q'$, and $n$ is the number of edges of the terrain.
\end{lemma}
\begin{proof}
Let $p$ be a neighbor of $q'$ and let $z_{\min}$ be $\max(\low(p), z')$. Obviously, $z_{\min}$
is a lower bound on the elevation that $p$ could have while still allowing flow
to $q'$.  There are three cases for the outcome of a query with $\widehat{q'}$
in the slope diagram of $p$.
\begin{compactenum}[(i)]
\item If $\widehat{q'}$ lies in the interior of $U(p)$, then
$q'$ can never be a steepest descent neighbor of $p$ in a non-ambiguous
realization. As such, $p$ is not included in the result of \expandPWSX{q'}{z'}.
\item If the line through $\widehat{q'}$ and $(0,z_{\min})$ does not intersect
the interior of $U(p)$, then we return $p$ with elevation
$z_{\min}$, unless $z_{\min} > \high(p)$.
\item Otherwise, we conduct a binary search on $Z(p)$ as indicated above to find
the lowest intersection $(0,z)$ of the vertical axis and a tangent of $U(p)$ through
$\widehat{q'}$. If $z > \high(p)$, we do not include $p$ in the result,
otherwise, we return $p$ with elevation $\max(z, \low(p))$.  Note that we do not
need to remove $q'$ itself from $U(p)$ (and $Z(p)$) in this procedure, since it
will never lose if it competes with itself.
\end{compactenum}
The computations can be done in time logarithmic in the degree of $p$.
\end{proof}

\subsubsection{Correctness and running time of the complete algorithm}

\begin{theorem}\thmlab{compute-pws}
After precomputations in $O(n \log n)$ time and $O(n)$ space,
the algorithm \\ \computePWSX{Q} 
computes the potential watershed $\PoWS(Q)$ of a
set of nodes $Q$ and its canonical realization $\CanonR(Q)$
in time $O(n \log n)$, where $n$ is the number of edges in the terrain.
\end{theorem}
\begin{proof}
The algorithm searches the graph starting from the nodes of $Q$. At each point
in time we have three types of nodes.  Nodes that have been extracted from the
priority queue have a \emph{finalized} elevation, a node that is currently in
the priority queue but was never extracted (yet) has a \emph{tentative}
elevation, other nodes have not been reached.

We show that when $(p,z)$ is first extracted from the priority queue in
\algref{compute-pws}, $p$ is indeed contained in the potential watershed of~$Q$,
and the elevation $z$ is the lowest possible  elevation of $p$ such that water
flows from $p$ to some node in $Q$. To this end we use an induction on the points
extracted, in the order in which they are extracted for the first time.

The induction hypothesis consists of two parts:
\begin{compactenum}[(i)]
\item
There exists a realization $R$ and $q \in Q$ such that $\elev(R,p) = z$, and
$R$ induces a flow path $\pi$ from $p$ to $q$ which only visits vertices
that have been extracted from the priority queue.

\item
There exists no realization $R$ and $q \in Q$ such that $\elev(R,p) < z$ and $p \flowsto{R} q$.
\end{compactenum}

If a node $q \in Q$ is extracted with $z = \low(q)$, then the claims hold
trivially. Note that the first extraction from the priority queue must be of
this type.

If $p$ is extracted from the priority queue for the first time and $p \notin Q$,
then there must be at least one node $p'$ that was extracted earlier, such that
\expandPWSX{p'}{z'}, for some elevation $z'$, resulted in $p$ having the
tentative elevation~$z$.
By induction, there exists a realization $R'$ and $q \in Q$, such that
$\elev(R',p')=z'$, there is a flow path $\pi$ from $p'$ to $q$ in $R'$, and
$\pi$ does not include~$p$.

To see (i), we construct a realization $R$ by modifying $R'$ as follows: we set
$\elev(R,p)=z$, and we set $\elev(R,r) = \high(r)$ for each neighbor $r$ of $p$
that does not lie on $\pi$. In comparison to $R'$, only $p$ and its neighbors
may have a different elevation in $R$. Since $\elev(R,p) = z \geq z'$ is still
at least as high as the elevation of any node on $\pi$, water will still flow
along the path $\pi$ from $p'$ to $q$. By the definition of \expandPWS, none of
the neighbors of $p$ that are set at their highest elevation can out-compete
$p'$ as a steepest-descent neighbour of $p$. Therefore, the steepest-descent
neighbour of $p$ in $R'$ must be one of the nodes on $\pi$. Thus, water from
$p$ will flow onto $\pi$, and thus, to~$q$.

\begin{wrapfigure}[8]{r}{0.3\textwidth}
\centering
\includegraphics{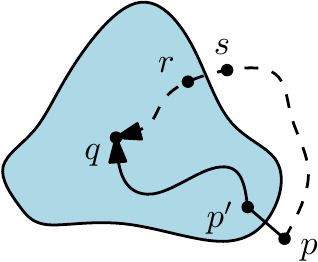}
\end{wrapfigure}

Next we show (ii).  Suppose, for the sake of contradiction, there is a realization $R$ such that $\elev(R,p) < z$ and there is a flow path from $p$ to a node $q \in Q$. Consider two consecutive nodes $r$ and $s$ on this path, such that $r$ has
not been extracted before but $s$ has been previously extracted (it may be that $r = p$ and/or $s \in Q$). Note that flow paths have to be monotone in the elevation. We argue that this path cannot stay below $z$ in any realization. Since $r$ is a neighbor of $s$, it has been added to the priority queue during the expansion of $s$. Let the tentative elevation of $r$ that resulted from this expansion be $z_r$.  By induction, since the elevation of $s$ is finalized, $z_r$ is a lower bound on the elevation of $r$ for any flow path that follows the edge $(r,s)$ and then continues to a node in $Q$ in any realization. However, $z_r \geq z$, since $r$ was not extracted from the priority queue before~$p$. Therefore, a path from $p$ to $q$ that contains $r$ with $\elev(R,p) < z$ cannot exist. This proves (ii).

It follows that the algorithm outputs all nodes of $\PoWS(Q)$ together with their elevations in $\CanonR(Q)$.

As for the running time, computing and storing $U(p)$ and $Z(p)$ for a node $p$
of degree $d$ takes $O(d \log d)$ time and $O(d)$ space. Since the sum of all
node degrees is $2n$, computing and storing $U(p)$ and $Z(p)$ for all nodes $p$ thus takes
$O(n \log d_{\max})$ time and $O(n)$ space in total, where $d_{\max}$ is the
maximum node degree in the terrain. While running algorithm \computePWSX{Q},
each node is expanded at most once. By \lemref{pre-expand-pws},
\expandPWSX{q'}{z'} on a node $q'$ of degree $d$ takes time $O(d \log
d_{\max})$. Thus, again using that all nodes together have total degree $2n$,
the total time spent on expanding is $O(n \log d_{\max}) = O(n \log n)$.  Each
extraction from the priority queue takes time $O(\log n)$ and there are at most
$O(n)$ nodes to extract.  Therefore \computePWS takes time $O(n\log n)$ overall.
\end{proof}

\bigskip
For grid terrains, $d_{\max} = O(1)$, and thus, the slope diagram computations
take only $O(1)$ time per expansion. In fact, since we only need to expand nodes
that are in $\PoWS(Q)$, we could actually compute $\PoWS(Q)$ in $O(k \log k)$
time, where $k = |\PoWS(Q)|$.
Alternatively, we can use the techniques from Henzinger et
al.~\cite{hkrs-fspapg-97} for shortest paths to overcome the priority queue
bottleneck, and obtain the following result (details in Appendix~\ref{sec:lineartime}):

\begin{theorem}\thmlab{compute-pws-grid}
The canonical realization of the potential watershed 
of a set of cells $Q$ in an
imprecise grid terrain of $n$ cells can be computed in $O(n)$ time.
\end{theorem}

\subsection{Potential downstream areas}
\seclab{deltas}

Similar to the potential watershed of a set $Q$, we can define the set of
points that potentially \emph{receive} water from a node in $Q$. Let
\[\PoDel(Q) =  \bigcup_{R \in \ReT} \bigcup_{q \in Q} \{ p : p \flowsto{R} q \}.\]

Naturally, a canonical realization for this set does not necessarily exist, however, it can be computed in a similar way as described in \secref{potential-ws} using a
priority queue that processes nodes in decreasing order of their maximal
elevation such that they could still receive water from a node in $Q$. The
algorithm is the same as \algref{compute-pws}, except that in the first line
the nodes are enqueued with their highest possible elevation, in line 3 we dequeue
the current node with the largest key and we use the following subroutine in
line 6.

\begin{defn}
Let $\expandPDX{q'}{z'}$ denote a function that returns for a node $q'$
and an elevation $z' \in [\low(q'),\high(q')]$ a set of pairs of nodes and elevations, which includes the pair $(p,z)$ if and only if $p \in N({q'})$, there is a realization $R$ with $\elev(R,q') \in [\low(q'),z']$ such that $q' \flowsto{R} p$, and $z$ is the maximum elevation of $p$ over all such realizations $R$.
\end{defn}

\begin{figure}[htb]\center\figlab{ch-expand-pd}
\includegraphics{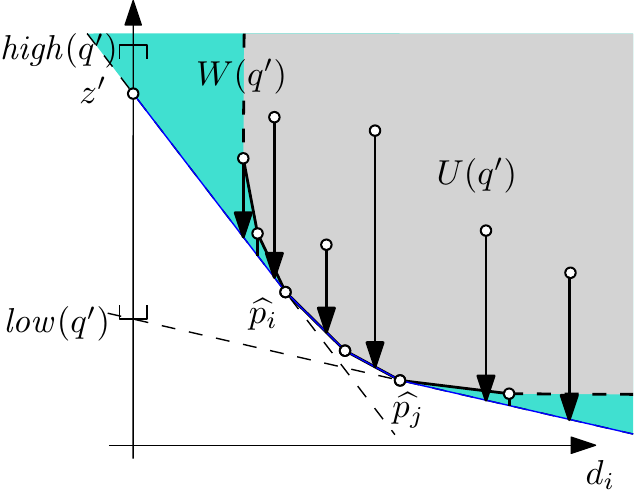}
\caption{Computations in the slope diagram}
\end{figure}

\begin{lemma}\lemlab{pre-expand-pd}
We can compute the function \expandPDX{q'}{z'} in $O(d \log d)$ time,
where $d$ is the node degree of $q'$.
\end{lemma}
\begin{proof}
Consider the slope diagram of $q'$ as defined in \secref{slope-diagram}.
Let $z_0$ be $\min \high(p)$ over all neighbours $p$ of $q'$; note that this is the vertical coordinate of the lowermost point of $U(q')$.
Let $\widehat{q'}=(0,z')$ and consider its lower tangent to $U(q')$. Let
$\widehat{p_i}$ be the corner of $U(q')$ that intersects the tangent.
Similarly, let $\widehat{p_j}$ be the corner of $U(q')$ that intersects the
tangent through $(0,\max(\low(q'),z_0))$.
Let $W(q')$ be the intersection of the halfplanes
above these two tangents and the halfplanes $H_i,\dots,H_j$ as defined in
\secref{slope-diagram}.
Clearly, a neighbor of $q'$ can have a steepest descent edge from $q'$,
for some elevation of $q'$ in $[\low(p),z']$, if and only if its representative in the slope diagram.
lies below $W(q')$ or on the boundary of $W(q')$.
To compute the neighbors of $q'$ and their elevations as they should be returned
by \expandPDX{q'}{z'}, we test each neighbor $p$ of $q'$ as follows.
We find the point $\widehat{p'} = (|pq'|, z)$ that is the projection from $\widehat{p}$ down onto the boundary of $W(q')$.
If $z \geq \low(p)$, we return $(p,z)$, otherwise we do not include $p$ in the result.

The slope diagram with $W(q')$ can be computed $O(d \log d)$ time. The neighbors
$p$ of $q'$ can be sorted by increasing distance from $q'$ in the
$xy$-projection in $O(d \log d)$ time; after that, the projections of all points
$\widehat{p}$ can be computed in $O(d)$ time in total by handling them in order
of increasing distance from $q'$ and walking along the boundary of $W(q')$
simultaneously.  \end{proof}



\begin{theorem}
Given a set of nodes $Q$ of an imprecise terrain, we can compute the set $\PoDel(Q)$ in time $O(n \log n)$, where $n$ is the number of edges in the terrain.
\end{theorem}
\begin{proof}
The algorithm searches the graph starting from the nodes of $Q$. As in the algorithm for potential watersheds, nodes that have been extracted from the priority queue have a \emph{finalized} elevation; nodes that are currently in the priority queue but were never extracted (yet) have \emph{tentative} elevations. However, this time these elevations are not to be understood as elevations of the nodes in a single realization, but simply as the highest known elevations so that the nodes may be reached from $Q$.

The induction hypothesis is symmetric to the hypothesis used for potential watersheds: we show that when $(p,z)$ is first extracted from the priority queue, $p$ is indeed contained in the potential downstream area of~$Q$, and the elevation $z$ is the highest possible elevation of $p$ such that water flows from some node in $Q$ to $p$. Again, the induction is on the points extracted, in the order in which they are extracted for the first time.

The induction hypothesis consists of two parts:
\begin{compactenum}[(i)]
\item
There exists a realization $R$ and $q \in Q$ such that $\elev(R,p) = z$, there is a flow path $\pi$ from $q$ to $p$ in $R$, and $\pi$ only visits vertices that have been extracted from the priority queue.

\item
There exists no realization $R$ and $q \in Q$ such that $\elev(R,p) > z$ and $q \flowsto{R} p$.
\end{compactenum}

If a node $q \in Q$ is extracted with $z = \high(q)$, then the claims hold trivially. Note that the first extraction from the priority queue must be of this type.

If $p$ is extracted from the priority queue for the first time and $p \notin Q$, then there must be at least one node $p'$ that was extracted earlier, such that \expandPDX{p'}{z'}, for some elevation $z'$, resulted in $p$ having the tentative elevation~$z$. By induction, there exists a realization $R'$ and $q \in Q$, such that $\elev(R',p')=z'$, there is a flow path $\pi$ from $q$ to $p'$ in $R'$, and $\pi$ does not include $p$.

So far the proof is basically symmetric to that of \thmref{compute-pws}. However, to see (i), we need a different construction. Let $z'' \leq z'$ be the elevation such that water flows from $p'$ to $p$ in the realization $R''$ with $\elev(R'',p') = z''$, $\elev(R'',p) = z$, and $\elev(R'',p'') = \high(p'')$ for all other nodes $p''$. Note that $z''$ exists by definition of \expandPD. We now construct a realization $R$ by modifying $R'$ as follows: we set $\elev(R,p') = z''$, we set $\elev(R,p) = z$, and we set $\elev(R,r) = \high(r)$ for each neighbor $r$ of $p'$ such that $r \neq p$ and $r$ does not lie on $\pi$. In comparison to $R'$, only two nodes in $R$ may have lower elevation, namely $p$ and $p'$. Therefore, water will still flow along the path $\pi$ from $q$ until it either reaches $p'$, or a vertex that now has $p$ or $p'$ as a new steepest-descent neighbour. Thus, in any case, there is a flow path from $q$ to either $p$ or $p'$. If the flow path reaches $p'$, then, by definition of \expandPD, none of the neighbours of $p'$ that are set at their highest elevation can out-compete $p$ as a steepest-descent neighbour of $p'$. Of course, the neighbours of $p'$ that lie on $\pi$ cannot out-compete $p$ either, since these neighbours have elevation at least as high as $p'$. Therefore, $p$ must be a steepest-descent neighbour of $p'$ in $R'$, and water from $p'$ will flow to $p$. Thus, in any case, water from $q$ will reach $p$ in $R'$ along a path that is a prefix of $\pi$, followed by an edge to $p$. This proves part (i) of the induction hypothesis.

The proof of part (ii) is completely analogous to the proof of \thmref{compute-pws}.

It follows that the algorithm outputs all nodes of $\PoDel(Q)$. The running time analysis is analogous to \thmref{compute-pws}.
\end{proof}

\subsection{Persistent watersheds}
\seclab{pers-ws}

In this section we will give a definition of a minimal watersheds, and explain how
to compute it.
Recall that the potential (maximal) watershed of a node set $Q$ is defined as the set of
nodes that have some flow path to a node in $Q$.
We can write this as follows
\[\PoWS(Q)= \left\{p: ~\exists~ \pi \in \FP(\ReT), \pi \ni p ~\exists~q \in Q: p
\flowsvia{\pi} q  \right\}.\]
An analogous definition to this would be
\[\CoWS(Q)= \left\{p: ~\forall~ \pi \in \FP(\ReT), \pi \ni p ~\exists~q \in Q: p
\flowsvia{\pi} q  \right\}.\]
This is the set of nodes $p$ from which water flows to $Q$ via \emph{any} induced
flow
path that contains~$p$.  We call this the \emphi{core watershed} of $Q$.

However, this definition seems a bit too restrictive. Consider the case of a measuring
device with a constant elevation error, used to sample points in a gently
descending valley. It is possible that, by increasing the density of measurement
points, we can create a region in which imprecision intervals of neighbouring
nodes overlap in the vertical dimension, and thus each node could become a local
minimum in some realization. Thus, water flowing down the valley could,
theoretically, ``get stuck'' at any point, and thus, the minimum watershed of
the point $q$ at the bottom of the valley would contain nothing but $q$ itself,
see \figref{silly:cws}.
\footnote{Interestingly, there are some parallels to observations made in the \textsc{gis}
literature. Firstly, Hebeler et al.~\cite{hebeler20094} observe that the watershed is
more sensitive to elevation error in ``flatlands''.  Secondly,
simulations have shown that also potential local minima or ``small sub-basins'' can
severely affect the outcome of hydrological computations~\cite{lindsay08}.}

\begin{figure}[tb]
\centering
\includegraphics{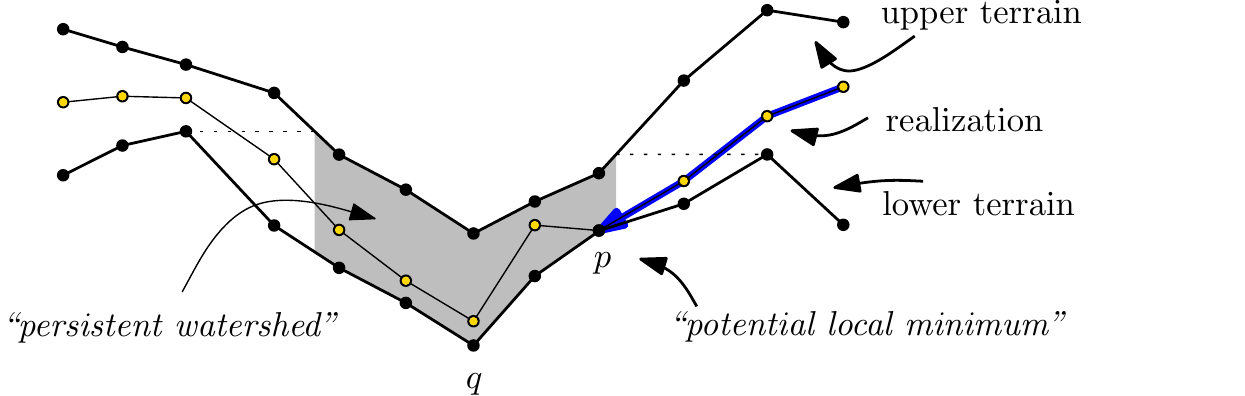}
\caption{An example of a 1.5 dimensional imprecise terrain, where the core
watershed can be arbitrarily reduced by oversampling. The node $p$ cannot be in the
core watershed of any other node.}
\figlab{silly:cws}
\end{figure}

Nevertheless, it seems clear that any water flowing in the valley must
eventually reach $q$ (possibly after flooding some local minima in the valley),
since the water has nowhere else to go.
This leads to an alternative definition of a minimal watershed, after we rewrite
the definition of the core watershed slightly.
Observe that the following holds for the complement of the core watershed.
\[\compl{\CoWS(Q)}= \left\{p: ~\exists~ \pi \in \FP(\ReT), \pi \ni p
~\lnot\exists~ q \in Q:  p \flowsvia{\pi} q  \right\}\]

Thus, the core watershed of $Q$ is the complement of the set of nodes $p$,
for which it is possible that water follows a flow path from $p$ that does \emph{not} lead to $Q$.
Assume there exists a suitable set of \emph{alternative
destinations} $S$, such that we can rewrite the above equation as follows
\[\compl{\CoWS(Q)}= \left\{p: ~\exists~ \pi \in \FP(\ReT), \pi \ni p ~\exists~s
\in S: (p \flowsvia{\pi} s) \wedge (\pi[p,s] \cap Q = \emptyset) \right\}.\]

Note that the right hand side is equivalent to the set
\begin{equation} \Eqlab{q:avoid:ws}
 \AvWS{Q}(S) := {\bigcup_{\pi \in \FP(\ReT) } \bigcup_{s \in S } \{p: (p \flowsvia{\pi} s)
\wedge (\pi[p,s] \cap Q = \emptyset) \}}
\end{equation}

We call the set in \Eqref{q:avoid:ws} the $Q$\emphi{-avoiding
potential watershed} of a set of nodes $S$ and we denote it with $\AvWS{Q}(S)$.
This is the set of nodes that have a potential flow path to a node $s \in S$ that does not
pass through a node of $Q$ before reaching $s$.
Note that it is still possible for those flow paths to intersect
$Q$, as long as this happens outside the subpath between $p$ and $s$.

It remains to identify the set of alternative destinations $S$.
Since every flow path ends in a local minimum, the set of potential local minima
clearly serves as such a set of destinations. 
Let $V^{\setminus Q}_{\min}$ be the union of all sets $r$ such that there exists
a realization in which all nodes of $r$ have the same elevation, $r$ is a local
minimum, and $r \cap Q = \emptyset$. 
However, it is also safe to include the nodes that do not have any flow path to $Q$,
which is the complement of the set $\PoWS(Q)$. 
It follows for the core watershed:


\[ \CoWS(Q) = \compl{~\AvWS{Q}(~V^{\setminus Q}_{\min} \cup \compl{~\PoWS(Q)} ~)~}\]

Note that we can rewrite this as follows:
\[ \CoWS(Q) = \compl{~\AvWS{Q}(~\compl{~\PoWS(Q)}~)~} ~\setminus~ \AvWS{Q}(~
V^{\setminus Q}_{\min} \cap  {\PoWS(Q)} ~) \]

Based on the above considerations we suggest the following alternative
definition of a minimal watershed.
\begin{defn}\deflab{persistent}
The \emphi{persistent watershed} of a set of nodes $Q$ is defined as
\[\PsWS(Q) := \compl{~\AvWS{Q}(~\compl{~\PoWS(Q)} ~)~} .\]
\end{defn}
The shaded area in \figref{silly:cws} indicates what would be the persistent watershed of $q$ in this case: these are the nodes that can never be high enough so that water from those nodes could escape from the potential watershed of $q$.

To compute the persistent watershed efficiently, all we need are efficient
algorithms to compute potential watersheds and $Q$-avoiding potential
watersheds.
We have already seen how to compute $\PoWS(Q)$ efficiently in
\secref{potential-ws}.  Note that the $Q$-avoiding potential watershed of $S$ is
different from the potential watershed of $S$ in the terrain $T'$ that is
obtained by removing the nodes $Q$ and their incident edges from $T$.  The next
lemma states that we can also compute $Q$-avoiding potential watersheds
efficiently.

\begin{lemma}\lemlab{compute-apws}
There is an algorithm which outputs the $Q$-avoiding potential watershed of $S$
and takes time $O(n \log n)$, where $n$ is the number of edges of the terrain.
\end{lemma}
\begin{proof}
We modify the algorithm to compute the potential watershed of $S$ as shown in
Algorithm~\ref{compute-pws}, such that, each time the algorithm extracts a node
from the priority queue, this node is discarded if it is contained in $Q$.
Instead, the algorithm continues with the next node from the priority queue.
Clearly, this algorithm does not follow any potential flow paths that flow
through $Q$. However, the nodes of $Q$ are still being considered by
the neighbors of its neighbors as a node they have to compete against for being
the steepest descent neighbor.  It is easy to verify that the proof of
\thmref{compute-pws} also holds for the computation of $Q$-avoiding potential
watersheds.
\end{proof}

\bigskip
By applying \thmref{compute-pws} and \lemref{compute-apws}, we obtain:

\begin{theorem}
We can compute the persistent watershed $\PsWS(Q)$ of $Q$ in time $O(n\log n)$, where $n$ is the number of edges of the terrain.
\end{theorem}

\section{Regular terrains}
\seclab{regular}

We extend the results on imprecise watersheds in the network model for a certain
class of imprecise terrains, which we call ``regular''. We will first define
this class and characterize it.  To this end we will introduce the notion of
imprecise minima (see \defref{fuzzybottom}), which are the ``stable'' minima of an imprecise terrain,
regular or non-regular.  In \secref{regularize} we will describe how to compute
these minima and how to turn a non-regular terrain into a regular terrain.
In the remaining sections, we discuss nesting properties and fuzzy boundaries of
imprecise watersheds. Furthermore, we observe that regular terrains have a
well-behaved ridge structure, that delineates the main watersheds.

The main focus of this section is on the extension of the results in section
\secref{network}. Some of the concepts introduced here could also be applied to
the surface model, however, we confine our discussion to the network model.

\subsection{Characterization of regular terrains}

We first give a definition of a proper minimum in an imprecise terrain.

\begin{defn}\deflab{fuzzybottom}
A set of nodes $S$ in an imprecise terrain $T$ is an \emphi{imprecise minimum}
if $S$ contains a local minimum in every realization of $T$, and no proper
subset of $S$ has this property.
\end{defn}

Now a regular imprecise terrain is defined as follows:

\begin{defn}\deflab{regular}
An imprecise terrain $T$ is a \emphi{regular imprecise terrain}, if every local
minimum of the lowermost realization $\Rlow$ of $T$ is an imprecise minimum of
$T$.
\end{defn}

Any imprecise minimum $S$ on a regular terrain is a minimum in $\Rlow$. Indeed,
if $S$ would not be a minimum on $\Rlow$, then, by \defref{fuzzybottom}, it
would contain a proper subset $S'$ that is a minimum on $\Rlow$ while $S'$ is
not an imprecise minimum of $T$ --~but this would contradict \defref{regular}. Now, we observe:

\begin{observation}
\obslab{flatbottom}
Let $S$ be an imprecise minimum on a regular terrain. Then each node $s \in S$
has the same elevation lower bound $\low(s)$. Furthermore, for each subset $S'
\subset S$ we have $\PoWS(S') = \PoWS(S)$ and $\WS(\Rlow,S') = \WS(\Rlow,S)$.
\end{observation}

%
%

We derive a characterization of imprecise minima. For this, we introduce proxies.

\begin{defn}\deflab{proxy}
A \emphi{proxy} of an imprecise minimum $S$ is a node $p \in S$, such that there are no realizations $R$ and nodes $q \notin S$ such that $p \flowsto{R} q$.
\end{defn}

Thus, water that arrives in a proxy of an imprecise minimum $S$, can never leave
$S$ anymore. This implies that the proxy is not in the potential watershed of
any set of nodes that lies entirely outside $S$. The following lemma states that
every imprecise minimum contains a proxy.

\begin{lemma}\lemlab{fb:characterization}
Let the \emphi{bar} of a set $S$ be $\clearance(S) = \min_{s \in S} \high(s)$.
A set $S$ is an imprecise minimum if and only if (i) $\clearance(S) < \min_{t \in N(S)} \low(t)$ and (ii) no proper subset $S'$ of $S$ has this property.
Every imprecise minimum has a proxy.
\end{lemma}
\begin{proof}
First observe that condition (i) implies that $S$ contains a local minimum in any realization.

If $S$ is an imprecise minimum, then, by definition, it contains a local minimum in any realization and no proper subset $S'$ of $S$ has this property. We argue that this implies (i) and (ii) for $S$.

To prove (i), consider the following realization $R$: For all nodes $r \in S$ we
set $\elev(R,r) = \max(\clearance(S), \low(r))$, and for all nodes $t \in N(S)$
we set $\elev(R,t) = \low(t)$. Now suppose, for the sake of contradiction, that
there is a proper subset $S'$ of $S$ that is a local minimum in $R$. Like all
nodes of $S$, the local minimum $S'$ must have elevation at least
$\clearance(S)$; each node $t \in N(S')$ must be set at a higher elevation
$\low(t)$. If we would remove the nodes of $N(S')$ from $S$, the imprecise
minimum $S$ would be separated into several components, including at least one component $S''$ that contains a node $s$ with $\high(s) = \clearance(S)$. This component $S''$ is a proper subset of $S$. Its neighbourhood $N(S'')$ consists of nodes from $N(S)$ and $N(S')$, all of which have an elevation lower bound strictly above $\clearance(S) = \min_{s \in S''} \high(s)$. Thus $S'' \subset S$ meets condition (i) and contains a local minimum in any realization, contradicting the assumption that $S$ is an imprecise minimum. If follows that no proper subset $S'$ of $S$ is a local minimum in $R$; therefore $S$ must be a local minimum as a whole, which implies (i).

To prove (ii), assume, for the sake of contradiction, that $S$ contains a proper subset $S'$ such that $\clearance(S') < \min_{t \in N(S')} \low(t)$. Thus, $S'$ would contain a local minimum in any realization, and $S$ would not be an imprecise minimum; hence (ii) must hold for $S$.

Now we argue that, if (i) and (ii) are met, then $S$ is an imprecise minimum. Recall that if condition (i) is met, then $S$ contains a local minimum in any realization. Now assume, for the sake of contradiction, that there exists a proper subset $S'$ that always contains a local minimum. Let $S'$ be a smallest such subset of $S$. We have that $S'$ is an imprecise minimum, and therefore, as we proved above, it holds that $\clearance(S') < \min_{t \in N(S')} \low(t)$, which contradicts that condition (ii) holds for $S$. Hence, there is no proper subset $S'$ of $S$ that always contains a local minimum; therefore $S$ is an imprecise minimum.

As a proxy of an imprecise minimum $S$, we take any node $s$ such that $\high(s)
< \min_{t \in N(S)} \low(t)$. By part (i) of the lemma, such a node $s$ always
exists. Since $s$ lies below any node of $N(S)$ in any realization, there are no
realizations $R$ and nodes $q \notin S$ such that $s \flowsto{R} q$; thus $s$ is
a proxy of $S$.  \end{proof}

\subsection{Computing proxies and regular terrains}
\seclab{regularize}

Any imprecise terrain can be turned into a regular imprecise terrain by raising
the lower bounds on the elevations such that local minima that violate the
regularity condition are removed from $\Rlow$. Indeed, in hydrological
applications it is common practice to preprocess terrains by removing local
minima before doing flow computations~\cite{tarboton1997new}. To do so while
still respecting the given upper bounds on the elevations, we can make use of
the algorithm from Gray et al.~\cite{gkls-rleit-10}. The original goal of this
algorithm is to compute a realization of a surface model that minimizes the
number of local minima in the realization, but the algorithm can also be
applied to a network model. It can easily be modified to output a
proxy for each imprecise minimum of a terrain. Moreover, the realization $M$
computed by the algorithm has the following convenient property: if we change
the imprecise terrain by setting $\low(v)$ to $\elev(M,v)$ for each node, we
obtain a regular imprecise terrain.

\paragraph{The algorithm}
The algorithm proceeds as follows. We will sweep a horizontal plane upwards.
During the sweep, any node is in one of three states. Initially, each node is
\emph{undiscovered}. Once the sweep plane reaches $\low(v)$, the state of the
node changes to \emph{pending}. Pending nodes are considered to be at the level
of the sweep plane, but they may still be raised further. During the sweep, we
will always maintain the connected components of the graph induced by the nodes
that are currently pending; we call this graph $G_P$. As soon as it becomes
clear that a node cannot be raised further or does not need to be raised
further, its final elevation on or below the sweep plane is decided and the node
becomes \emph{final}. More precisely, the algorithm is driven by two types of
events: we may reach $\low(v)$ for some node $v$, or we may reach $\high(v)$ for
some node $v$. These events are handled in order of increasing elevation;
$\low(v)$-events are handled before $\high(v)$-events at the same elevation. The
events are handled as follows:\begin{itemize}
\item reaching $\low(v)$: we make $v$ pending, and find the component $S$ of $G_P$ that contains $v$. If $v$ has a neighbour that is final, we make all nodes of $S$ final at elevation $\low(v)$.
\item reaching $\high(v)$: if $v$ is final, nothing happens; otherwise we report  $v$ as a proxy, we find the connected component $S$ of $G_P$ that contains $v$, and we make all nodes of $S$ final at elevation\footnote{This is a small variation: the algorithm as described originally by Gray et al.\ would make the elevations final at $\high(v) = \min_{s \in S} \high(s)$. However, in the current context we prefer to make the elevations final at $\max_{s \in S} low(s)$, to maintain as much of the imprecision in the original imprecise terrain as possible.} $\max_{s \in S} \low(s)$.
\end{itemize}
Gray et al.\ explain how to implement the algorithm to run in $O(n \log n)$
time, ~\cite{gkls-rleit-10}.

\begin{lemma}
\lemlab{findtheproxies}
Given an imprecise terrain $T$, (i) all nodes reported by the above algorithm are proxies of imprecise minima, and (ii) the algorithm reports exactly one proxy of each imprecise minimum of $T$.
\end{lemma}
\begin{proof}
We first prove the second part, and then the first part of the lemma.

(ii) Let $S$ be an imprecise minimum. Let $v$ be the node in $S$ which was the
first to have its $\high(v)$-event processed. By \lemref{fb:characterization},
$v$ is a proxy of $S$ and we have $\high(v) < min_{t \in N(S)} \low(t)$.  Hence,
when $\high(v)$ is processed, the component of $G_P$ that contains $v$ does not
contain any nodes outside $S$, and the $\high(v)$-event is the first event to
make any nodes in this component final. Thus, $v$ is reported as a proxy.
Furthermore, no node $s \in S$ can have $\low(s) > \high(v)$, otherwise
$\clearance(S \setminus \{s\}) = \high(v) < \min_{t \in \{s, N(S)\}} \low(t) \leq \min_{t \in N(S \setminus \{s\})} \low(t)$,
and thus, by
\lemref{fb:characterization}, $S$ would not be an imprecise minimum. Hence, when
the $\high(v)$-event is about to be processed, all nodes of $S$ have been
discovered and are currently pending. The $\high(v)$-event makes all nodes of
$S$ final; thus, any $\high(s)$-events for other nodes $s \in S$ will remain
without effect and no more proxies of $S$ will be reported.

(i) Let $v$ be a node that is reported as a proxy in a $\high(v)$-event. We
claim that the connected component $S$ of $G_P$ that contains $v$ at that time,
is an imprecise minimum. Indeed, by definition of $G_P$, all nodes of $S$ are
pending, and thus $\high(v) = \min_{s \in S} \high(s) = \clearance(S)$.
Furthermore, because $S$ is a connected component of $G(P)$, all nodes $t \in
N(S)$ must be either undiscovered or final. In fact, the algorithm maintains the
invariant that no neighbor of a finalized node is pending; since all nodes in
$S$ are pending, all nodes $t \in N(S)$ must be undiscovered. Therefore
$\high(v) \leq \min_{t \in N(S)} \low(t)$. Because all $\low(t)$-events at the
same elevation as $\high(v)$ are processed before the $\high(v)$-event is
processed, we actually have a strict inequality: $\high(v) < \min_{t \in N(S)}
\low(t)$. It follows that $S$ satisfies condition (i) of
\lemref{fb:characterization}. Furthermore, no proper subset $S'$ of $S$ has this
property, otherwise, by the analysis given above, a proxy for $S'$ would have
been reported already and the nodes from $S'$ would have been removed from $G_P$
at that time. Hence, $S$ also satisfies condition (ii) of
\lemref{fb:characterization}, and $S$ is an imprecise minimum, with $v$ as a
proxy.  \end{proof}

\begin{lemma}
\lemlab{regularize}
Let $M$ be the realization of a terrain $T$ as computed by the algorithm described above. Let $T'$ be the imprecise terrain that is obtained from $T$ by setting $\low(v) = \elev(M,v)$ for each vertex $v$. The terrain $T'$ is a regular imprecise terrain.
\end{lemma}
\begin{proof}
Note that $M$ is the lowermost realization of $T'$. Consider any local minimum $S$ of $M$. Observe that the algorithm cannot have finalized the elevations of the last pending vertices of $S$ in a $\low(v)$-event, because then we would have $v \in S$ and $v$ must have a neighbor $t \notin S$ that was finalized before $v$; hence $\elev(M,t) \leq \elev(M,v)$ and $S$ would not be a local minimum. Therefore, the algorithm must have finalized the last elevations of the vertices of $S$ in a $\high(v)$-event for a vertex $v \in S$. Furthermore, each vertex $t \in N(S)$ must have been undiscovered at that time; otherwise $t$ would have become part of the same component as the vertices of $S$ before its elevations were finalized, or $t$ would have been finalized before $v$: in both cases $S$ would not be a local minimum. Hence we have $\low(t) > \high(v)$ for each vertex $t \in N(S)$, and thus, $S$ is a local minimum in every realization of $T$ or $T'$. Furthermore, no proper subset of $S'$ of $S$ contains a local minimum in every realization of $T'$, since in particular, in $M$ the set $S$ is a local minimum and therefore no proper subset $S'$ of $S$ is a local minimum. Thus, by \defref{fuzzybottom} and \defref{regular}, $T'$ is a regular terrain.
\end{proof}


\subsection{Nesting properties of imprecise watersheds}
\seclab{nesting}
To be able to design data structures that store imprecise watersheds and answer
queries about the flow of water between nodes efficiently, it would be convenient if
the watersheds satisfy the following \emphi{nesting condition}: if $p$ is
contained in the watershed of~$q$, then the watershed of $p$ is contained in the
watershed of~$q$.
Clearly, potential watersheds do not satisfy this nesting condition, while core
watersheds do.
However, in general, persistent watersheds, are not nested in this way. We give
a counter-example that uses a non-regular terrain in the next lemma before
proving the nesting condition for persistent watersheds in regular terrains later in this section.

\begin{lemma}\lemlab{counter-nesting}
There exists an imprecise terrain with two nodes $p$ and $q$
such that $p \in \PsWS(q)$ and $\PsWS(p) \nsubseteq \PsWS(q)$.
\end{lemma}
\begin{proof}
We give an example of a non-regular terrain that has this property.  Refer to
\figref{cws-counter}.
The persistent watershed of $p$ as shown in red is not completely contained in
the persistent watershed of $q$ as shown in blue.
The left figure gives a top-view. All edges have unit length, except for the edge between $w$ and $q$.
The right figure shows the fixed elevations of $s, t, t', u, v$ and $w$, the
elevation intervals of $p$, $q$ and $r$, and the correct horizontal distances on all edges except $|pv|$ and $|qv|$.
The red outline delimits $\PoWS(p) = \{p,q,r,s,t,v,w\}$.
The red dashed outline delimits $\PsWS(p) = \{p,s,v\}$.
The blue outline delimits $\PoWS(q) = \{p,q,s,v,w\}$.
The blue dashed outline delimits $\PsWS(q) = \{p,q,v\}$.
\end{proof}

\begin{figure}[htb]\center
\includegraphics[scale=1]{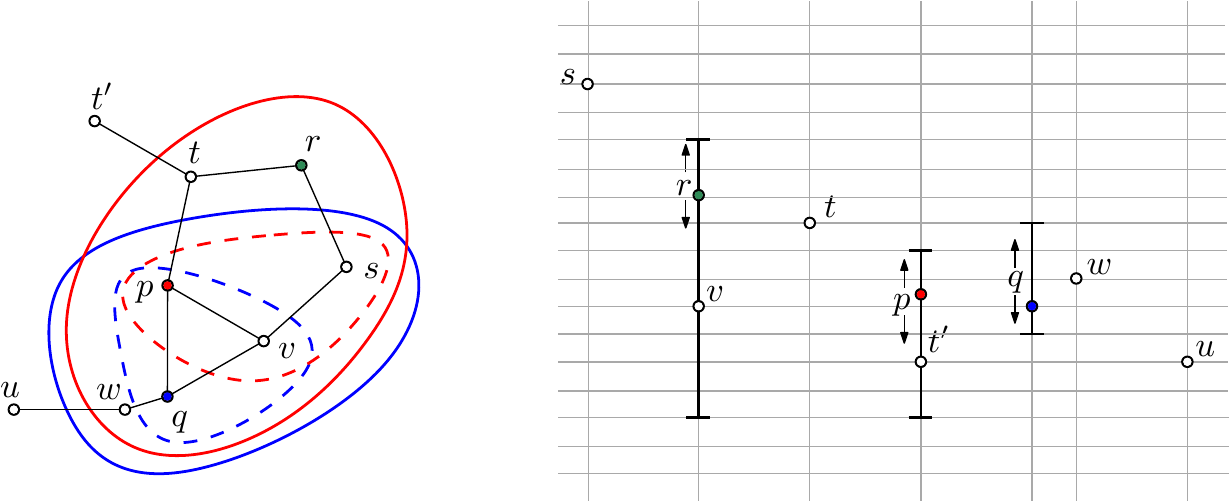}
\caption{Counterexample of a non-regular terrain to the nesting condition of persistent watersheds.}
\figlab{cws-counter}
\end{figure}


The following lemmas will prove that on regular imprecise terrains persistent
watersheds do satisfy the nesting condition.

\begin{lemma}\lemlab{pws:low}
Let $Q$ be a set of nodes in a regular imprecise terrain, then $\PersistentWS(Q) \subseteq \WS(\Rlow,Q)$.
\end{lemma}
\begin{proof}
Consider a flow path from a node $p \in \PersistentWS(Q)$ in $\Rlow$. By the definition of persistent watersheds, the path cannot leave $\PoWS(Q)$ without going through $Q$.

If the path reaches a local minimum $S$ without going through any node of $Q$,
we claim that this local minimum must contain a node $q \in Q$. To prove this
claim, assume, for the sake of contradiction, that there is a $Q$-avoiding flow
path to a local minimum $S$ such that $Q \cap S = \emptyset$. Thus there
would be a $Q$-avoiding flow path to any node of $S$, and in particular, to a
proxy $s \in S$; because the terrain is regular, $S$ must be an imprecise
minimum (by \defref{regular}), and therefore a proxy is guaranteed to exist by
\lemref{fb:characterization}. As observed above, $s$ is not in the potential
watershed of any set of nodes outside $S$; in particular, $s$ is not in
$\PoWS(Q)$. This implies that there is a realization in which a flow path from
$p$ leaves $\PoWS(Q)$ without going through any node of $Q$, contradicting the
assumption that $p \in \PsWS(Q)$. Hence, if a flow path from $p$ in $\Rlow$
reaches a local minimum $S$ without going through any node of $Q$, then $S$ must
contain a node $q \in Q$, and there would also be a flow path from $p$ to $q$.

Therefore, from any node $p \in \PersistentWS(Q)$ there must be a flow path to a
node $q \in Q$ in $\Rlow$, and thus, $\PersistentWS(Q) \subseteq \WS(\Rlow,Q)$.
\end{proof}

\begin{lemma}\lemlab{nesting:ws:b}\lemlab{plunging}
Let $Q$ be a set of nodes on an imprecise terrain, and let $P \subseteq
\WS(\Rlow,Q)$. Then $\PoWS(P) \subseteq \PoWS(Q)$.
\end{lemma}
\begin{proof}
Let $\overline{R}$ be the watershed-overlay of $\WS(\CanonR(P),P)$ and $\WS(\Rlow,Q)$.
Consider a node $r \in \PoWS(P)$ and a flow path $\pi$ from $r$ to a node $p \in P$ in $\CanonR(P)$.
Let $\pi'$ be the maximal prefix of $\pi$ such that the nodes of $\pi'$ have the same elevation in $\CanonR(P)$ and $\overline{R}$, and let $\pi''$ be the maximal prefix of $\pi'$ such that $\pi''$ is still a flow path in $\overline{R}$. We distinguish three cases:\begin{itemize}
\item If $\pi' = \pi''$ is empty, then $r$ has lower elevation in $\CanonR(P)$ than in $\overline{R}$, so $r$ must be in $\WS(\Rlow,Q)$.
\item If $\pi' = \pi'' = \pi$, then flow from $r$ reaches a node $p \in P \subseteq \WS(\Rlow,Q)$ in $\overline{R}$.
\item Otherwise, let $(u,v)$ be the edge of $\pi$ such that $u$ is the last node of $\pi''$. Now $v$ is not on $\pi''$, so in $\overline{R}$, flow from $u$ either still follows $(u,v)$ but $\elev(\overline{R},v) < \elev(\CanonR(P),v)$, or flow from $u$ is diverted over an edge $(u,\widehat{v})$ to another node $\widehat{v}$ with $\elev(\overline{R},\widehat{v}) < \elev(\CanonR(P),\widehat{v})$. In either case, from $u$ we follow an edge to a node of which the elevation in $\overline{R}$ is lower than in $\CanonR(P)$; therefore this must be a node of $\WS(\Rlow,Q)$.  \end{itemize}
In any case, there is a flow path from $r$ to a node of $\WS(\Rlow,Q)$. From here, there must a path to a node $q \in Q$, since every flow path within $\WS(\Rlow,Q)$ in $\Rlow$ is also a flow path in $\overline{R}$. Thus there is flow path from $r$ to $q$ in $\overline{R}$, and thus, $r \in \PoWS(Q)$. This proves the lemma.
\end{proof}

\begin{figure}[tb]\center
\includegraphics{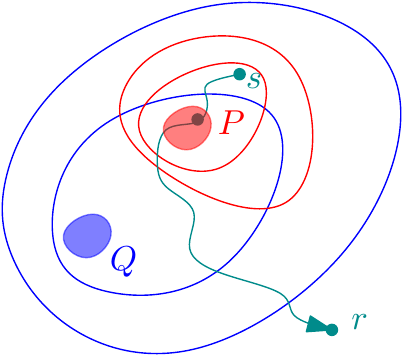}
\caption{Persistent watersheds are nested on regular terrains. Illustration to
the proof of \lemref{nesting-ws}. }
\figlab{nesting}
\end{figure}

\begin{lemma}\lemlab{nesting-ws}
(persistent watersheds are nested) Let $Q$ be a set of nodes in a regular imprecise terrain, and let $P \subseteq
\PsWS(Q)$. Then $\PsWS(P) \subseteq \PsWS(Q)$.  \end{lemma}
\begin{proof}
Assume for the sake of contradiction that there exists a node $s \in \PsWS(P)$,
such that $s \notin \PsWS(Q)$. Clearly, $s \in \PoWS(P)$ and by
\lemref{nesting:ws:b} and \lemref{pws:low}, it holds that $s \in \PoWS(Q)$.
Furthermore, $s$ must have a flow path $\pi$ to a point $r \notin \PoWS(Q)$, which
does not pass through a node of $Q$, refer to \figref{nesting}.
By \lemref{nesting:ws:b} and  \lemref{pws:low}, $\PoWS(P) \subseteq \PoWS(Q)$,
and thus $r \notin \PoWS(P)$.  Furthermore, since $s \in \PsWS(P)$, the subpath
$\pi[s,r]$ must include a node $p \in P$. This contradicts the fact that $p  \in
\PsWS(Q)$, since $\pi[p,r] \cap Q = \emptyset$.
\end{proof}

\subsection{Fuzzy watershed boundaries}
\seclab{fuzzy-boundary}

\lemref{pws:low} and \lemref{plunging} also allow us to compute the difference
between the potential and a persistent watershed of a set of nodes $Q$
efficiently, given only the boundary of the watershed of $Q$ on the lowermost
realization of the terrain. We first define these concepts more precisely.

\begin{defn}\deflab{fuzzyboundary}
Given a realization $R$, and a set of nodes $Q$, let $\crossing(R,Q)$ be the
directed set of edges $(u,v)$ such that $u \in \WS(R,Q)$ and $v \notin
\WS(R,Q)$. We call $\crossing(R,Q)$ the \emphi{watershed boundary} of $Q$ in
$R$.
Likewise, we define the \emphi{fuzzy watershed boundary} of $Q$
as the directed set of edges $(u,v)$ such that $u \in \PoWS(Q)$ and $v \notin
\PsWS(Q)$ and we denote it with $\PoBound(Q)$.
We call the set $\PoWS(Q) \setminus \PsWS(Q)$ the \emphi{uncertainty area} of
this boundary.
\end{defn}

We will now discuss how we can compute the uncertainty area of any fuzzy watershed boundary efficiently.

\paragraph{Algorithm to compute the uncertainty area of a watershed.}
Assume we are given $\crossing(\Rlow,Q)$.
We will compute the set $\PoWS(Q) \setminus \PsWS(Q)$ with \algref{compute-pws},
modified as follows. Instead of initializing the priority queue with the nodes
of $Q$, we initialize in the following way. For each edge $(u,v) \in \crossing(\Rlow,Q)$, we
use the slope diagram of $u$ to determine the minimum elevation $z_u$ of $u$,
such that there is a realization in which water flows on the edge from $u$ to
$v$. If there exists such an elevation $z_u$, we enqueue $u$ with elevation (and
key) $z_u$. Similarly, we use the slope diagram of $v$ to determine the minimum
elevation $z_v$ of $v$, such that water may flow on the edge from $v$ to $u$. If
$z_v$ exists, we enqueue $v$ with elevation (and key) $z_v$. After initializing
the priority queue in this way, we run \algref{compute-pws} as written.

\begin{lemma}\lemlab{fuzzyboundary}
If the terrain is regular, the algorithm described above computes $\PoWS(Q) \setminus \PsWS(Q)$.
This is the uncertainty area of the fuzzy watershed boundary of $Q$.
\end{lemma}
\begin{proof}
Observe, following the proof of \thmref{compute-pws}, that for any node $p$ output by the above algorithm, there are a realization $R_s$ and a node $s$ which was in the initial queue with elevation $z_s$, such that $\elev(R_s,s) = z_s$ and $R_s$ induces a flow path $\pi$ from $p$ to $s$. Let $(u,v)$ be the edge of $\crossing(\Rlow,Q)$ which led to the insertion of $s \in \{u,v\}$ into $Q$ with elevation $z_s$. Let $t$ be the other node of $(u,v)$, that is, $t = \{u,v\} \setminus \{s\}$. let $R_t$ be the realization obtained from $R_s$ by setting $\elev(R_t,t) = \low(t)$. Observe that, by our choice of $z_s$, the realization $R_t$ now induces a flow path from $p$ to $t$. We will now argue that (i) $p \in \PoWS(Q)$, and (ii) $p \notin \PsWS(Q)$.

(i) The existence of $R_u$ implies that $p \in \PoWS(u)$; since $u \in
\WS(\Rlow,Q)$ (by definition of $\crossing(\Rlow,Q)$ this implies $p \in \PoWS(Q)$ (by \lemref{plunging}).

(ii) By definition of $\crossing(\Rlow,Q)$, there is no flow path from $v$ to $Q$ on
$\Rlow$. Hence, any flow path from $v$ on $\Rlow$ must lead to a local minimum
$S$ that does not contain any node of $Q$, and by \defref{regular}, each such
local minimum $S$ is an imprecise minimum. Now, by \lemref{fb:characterization},
each such local minimum $S$ contains a proxy $s$, which is, by \defref{proxy},
not contained in $\PoWS(Q)$. Thus there is a flow path from $v$ that does not go
through any nodes of $Q$ and leads to a proxy $s \notin \PoWS(Q)$. Hence, by
\defref{persistent}, $p \notin \PsWS(Q)$.

Next, we will argue that if $p \in \PoWS(Q)$ and $p \notin \PsWS(Q)$, the algorithm will output $p$. We distinguish two cases.

If $p \in \WS(\Rlow,Q)$, then, because $p \notin \PsWS(Q)$, there must be a flow path on $\Rlow$ from $p$ to a minimum $S$ that does not contain any node of $Q$. By \defref{regular}, \lemref{fb:characterization} and \defref{proxy}, there will then be a flow path from $p$ to a proxy $s \in S$ that lies outside $\PoWS(Q)$, and thus, outside $\WS(\Rlow,Q)$.

If $p \notin \WS(\Rlow,Q)$, then, because $p \in \PoWS(Q)$, there must be a realization in which there is a flow path from $p$ to $Q$, and thus, from $p$ to $\WS(\Rlow,Q)$.

In both cases, there is a realization in which there is a flow path from $p$
that traverses an edge $(u,v) \in \crossing(\Rlow,Q)$, either from $u$ to $v$ or from $v$ to $u$. The algorithm reports at least all such points $p$.

This completes the proof of the lemma.
\end{proof}

\bigskip
Note that if all nodes have degree $O(1)$, the running time of the above
algorithm is linear in the size of the input ($\crossing(\Rlow,Q)$) and the
output ($\PoWS(Q) \setminus \PsWS(Q)$). When a data structure is given that
stores the boundaries of watersheds on $\Rlow$ so that they can be retrieved
efficiently, and the imprecision is not too high, this would enable us to compute
the boundaries and sizes of potential and persistent watersheds much faster
than by computing them (or their complements) node by node with
\algref{compute-pws}.

We can use the same idea as above to compute an uncertain area of the watershed boundaries between a set of nodes $Q$. More precisely, given a collection of nodes $Q$ such that no node $q \in Q$ is contained in the potential watershed of another node $q' \in Q$, we can compute the nodes that are in the potential watersheds of multiple nodes from $Q$.

\paragraph{Algorithm to compute the uncertainty area between watersheds.}
Let $Q$ be $\{q_1,...,q_k\}$ and let $G'$ be the graph induced by the potential watershed of $Q$. The algorithm is essentially the same as algorithm that computes the uncertainty area of a single watershed's boundary---the main difference is that now we have to start it with a suitable set of edges ${\cal X}$ on the fuzzy boundaries \emph{between} the watersheds of the nodes of $Q$. More precisely, ${\cal X}$ should be an edge separator set of $G'$, which separates the nodes of $G'$ into $k$ components $G'_1,...,G'_k$ such that nodes of each component $G'_i$ are completely contained in $\PoWS(q_i)$.

We obtain ${\cal X}$ with the following modification of \algref{compute-pws}. For each node $p$ we will maintain, in addition to a tentative elevation $z$, a tentative tag that identifies a node $q \in Q$ such that there is a realization $R$ with $\elev(R,p) = z$ and $p \flowsto{R} q$. We initialize the priority queue of \algref{compute-pws} with all nodes $q \in Q$, each with tentative elevation $\low(z)$ and each tagged with itself. The first time any particular node $q'$ is extracted from the priority queue, we obtain not only its final elevation but also its final tag $q$ from the queue, and each pair $(p,z) \in \expandPWSX{q'}{z'}$ is enqueued with that same tag $q$.
At the end of \algref{compute-pws}, we obtain the set of nodes in $\PoWS(Q)$ together with their elevations in the canonical realization $\CanonR(Q)$ and with tags, such that any set of nodes tagged with the same tag $q \in Q$ forms a connected subset of $\PoWS(q)$. We now extract the separator set ${\cal X}$ by identifying the edges between nodes of different tags.

Having obtained ${\cal X}$, we compute the union of the pairwise intersections of the potential watersheds of $q_1,...,q_k$ as follows. Again, we use \algref{compute-pws}. This time the priority queue is initialized as follows. For each edge $(u,v) \in {\cal X}$, we use the slope diagram of $u$ to determine the minimum elevation $z_u$ of $u$, such that there is a realization $R$ \emph{with $\elev(R,v) = \elev(\CanonR(Q),v)$} in which water flows on the edge from $u$ to $v$. If there exists such an elevation $z_u$, we enqueue $u$ with elevation (and
key) $z_u$. Similarly, we use the slope diagram of $v$ to determine the minimum
elevation $z_v$ of $v$, such that water may flow on the edge from $v$ to $u$ \emph{at elevation $\elev(\CanonR(Q),u)$}. If $z_v$ exists, we enqueue $v$ with elevation (and key) $z_v$. After initializing the priority queue in this way, we run \algref{compute-pws} as written, and output the result.

\begin{lemma} \lemlab{all:fuzzy:boundaries}
Given a set of nodes $q_1,\dots,q_k$ of an imprecise terrain, such that $q_i
\notin \PoWS(q_j)$ for any $i \neq j$ and $1\leq i,j\leq k$, we can compute the
set $\bigcup_i \bigcup_{j \neq i} (\PoWS(q_i) \cap \PoWS(q_j))$ in $O(n \log n)$
time, where $n$ is the number of edges of the imprecise terrain.
\end{lemma}
\begin{proof}
The separator set ${\cal X}$ is obtained in $O(n \log n)$ time by running the modified version of \algref{compute-pws} and one scan over the graph to identify edges between nodes with different tags. Computing the union of the pairwise intersections of the potential watersheds of $q_1,...,q_k$ with the modified \algref{compute-pws} takes $O(n \log n)$ time again.

By the same arguments as in the proof of \lemref{fuzzyboundary}, we can observe the following: for any node $p$ output by the above algorithm, there is an edge $(u,v) \in {\cal X}$, a realization $R_u$ with $\elev(R_u,u) = \elev(\CanonR(Q),u)$ and $p \flowsto{R_u} u$, and a realization $R_v$ with $\elev(R_v,v) = \elev(\CanonR(Q),v)$ and $p \flowsto{R_v} v$. Let $q_u, q_v \in Q$ be the nodes of $Q$ with which $u$ and $v$ were tagged, respectively. It follows that there is a flow path from $p$ to $q_u$ in the watershed overlay of $\WS(R_u,u)$ and $\WS(\CanonR(Q),q_u)$, so $p \in \PoWS(q_u)$. Analogously, $p \in \PoWS(q_v)$. Since $(u,v) \in {\cal X}$, we have $q_u \neq q_v$, so any point $p$ that is output by the algorithm lies in the intersection of the potential watersheds of two different nodes from $Q$.

Next, we will argue that if $p$ lies in the intersection of the potential watersheds of two different nodes from $Q$, then the algorithm will output $p$.
Let $q \in Q$ be the node with which $p$ is tagged (hence, $p \in \PoWS(q)$), and let $q' \in Q, q' \neq q$ be another node from $Q$ such that $p \in \PoWS(q')$. Consider a flow path $\pi$ from $p$ to $q'$ in $\CanonR(q')$, and let $(r,r')$ be the edge on $\pi$ such that $r$ is tagged with a node other than $q'$ and all nodes of $\pi[r',q']$ are tagged with $q'$. Note that $(r,r')$ must exist because all nodes of $\pi$ lie in $\PoWS(Q)$ and have received a tag, $p$ is tagged with another node than $q'$, and, since none of the nodes of $Q$ lie in each other's potential watersheds, $q'$ is tagged with itself. Therefore $(r,r')$ exists, and $(r,r') \in {\cal X}$. Moreover, we have $\elev(\CanonR(Q),r') = \elev(\CanonR(q'),r')$. Therefore $r$ was put in the priority queue with the minimum elevation such that there is a realization $R$ with $\elev(R,r') = \elev(\CanonR(q'),r')$ in which water flows on the edge from $r$ to $r'$. By induction on the nodes of $\pi$ from $r$ back to $p$, it follows that $p$ must eventually be extracted from the priority queue and output.

This completes the proof of the lemma.
\end{proof}

\subsection{The fuzzy watershed decomposition}
\seclab{fuzzy-ridge}

In this section we further characterize the structure of imprecise terrains by
considering the ridge lines that delineate the ``main'' watersheds.  In fact,
the fuzzy watershed boundaries (\defref{fuzzyboundary}) of the imprecise minima
(\defref{fuzzybottom}) possess a well-behaved ridge structure if the terrain is
regular.  Consider the following definition of an ``imprecise'' ridge.

\begin{defn}\deflab{fuzzy:ridge}
Let $S_1,..S_k$ be the imprecise minima of an imprecise terrain. We call the
union of the pairwise intersection of the potential watersheds of imprecise minima
the \emphi{fuzzy ridge} of the terrain.
\end{defn}

Let $S$ be an imprecise minimum of a regular imprecise terrain.
The next lemma testifies that the persistent watershed of any proxy $q$ of $S$
is equal to the intersection of the persistent watersheds of all possible
subsets of $S$.
Therefore, we think of $\PsWS(q)$ as the actual minimal watershed of $S$, or
the minimum associated with $S$.
By \obsref{flatbottom}, the potential watersheds of all subsets of $S$ are equal.
Consequently, we think of the fuzzy watershed boundary of $q$ as the fuzzy
watershed boundary of $S$.

\begin{lemma}\lemlab{proxy:persws}
Let $S$ be an imprecise minimum on a regular terrain, and let $x$ be any proxy
of $S$. Then $\bigcap_{Q \subseteq S} \PersistentWS(Q) = \PersistentWS(x)$.
\end{lemma}
\begin{proof}
We want to argue about the intersection of the persistent
watersheds of all subsets of $S$. Consider the complement $C$ of this set,
\[C := \compl{\bigcap_{Q \subseteq S} \PersistentWS(Q)} =
\bigcup_{Q \subseteq S} \compl{\PersistentWS(Q)} = \bigcup_{Q
\subseteq S} \AvWS{Q}(\compl{\PotentialWS(Q)}).\]

By \obsref{flatbottom} we have $\PoWS(Q)=\PoWS(x)=\PoWS(S)$ for any $Q
\subseteq S$, so we have $C = \cup_{Q \subseteq S} \AvWS{Q}(\compl{\PotentialWS(S)})$.
Now, the nodes contained in $C$ can be characterized as follows.
For any node $p \in C$, there must be a realization, in which there is a subset $S' \subseteq S$ and a node $q$ outside $\PotentialWS(S)$, such that there is a flow path from $p$ to $q$ that does not contain any node of $S'$. The given node $\{x\}$ always serves as such a subset $S'$ that is being ``avoided'', since $x$ is a proxy and, by \defref{proxy}, it is impossible for water that reaches $x$ to continue to flow to a node outside of $\PotentialWS(S)$. Therefore,
\[  \AvWS{Q}(\compl{\PotentialWS(S)})
= \PotentialWS(\compl{\PotentialWS(S)})
= \AvWS{x}(\compl{\PotentialWS(x)}) = C.
\]
The claim now follows from the definition of persistent watersheds.
\end{proof}

\vspace{\baselineskip}
We can now further characterize the fuzzy ridge for
regular terrains. The following lemma implies that on a regular terrain, the
fuzzy ridge is equal to the union of the uncertainty areas of the fuzzy watershed
boundaries of any representative set of proxies of the imprecise minima, see
\corref{main:structure}.

\begin{lemma} \lemlab{fuzzy:ridge}
Let $S_1,..S_k$ be the imprecise minima of a regular imprecise
terrain and let $q_1,..q_k$ be associated proxies. For any $1\leq i\leq k$, we have that
$\PersistentWS(q_i) = \compl{ \bigcup_{j\neq i} \PotentialWS(q_j)}$.
\end{lemma}
\begin{proof}
Since $q_i$ is a proxy, we have that,
\[\compl{\PersistentWS(q_i)} = \AvWS{q_i}(\compl{\PotentialWS(q_i)})
 = \PotentialWS(\compl{\PotentialWS(q_i)})
 = \bigcup_{p \in \compl{\PotentialWS(q_i)}} \PotentialWS(p).
\]

Now, for a node $p \in \compl{\PotentialWS(q_i)}$, consider a minimum $S$ that
is reached by a flow path from $p$ in $\Rlow$.  By \defref{regular}, we have
that $S$ is an imprecise minimum, and since $p \in \compl{\PoWS(q_i)} =
\compl{\PoWS(S_i)}$, we have $S \neq S_i$. As such, $S$ must be equal to some
$S_j$ for $j \neq i$.  Furthermore, by \obsref{flatbottom} we have
$\WS(\Rlow,S_j) = \WS(\Rlow,q_j)$, therefore $p \in \WS(\Rlow,q_j)$. Now,
\lemref{nesting:ws:b} implies that $\PoWS(p) \subseteq \PoWS(q_j)$. It follows
that $\compl{\PersistentWS(q_i)} \subseteq \bigcup_{i \neq j} \PotentialWS(q_j).
$

Since we also have that $q_j \in \compl{\PoWS(q_i)}$ for any $j \neq i$, we also get
$\compl{\PersistentWS(q_i)}
 \supseteq \bigcup_{i \neq j} \PotentialWS(q_j),$ which implies the equality.
\end{proof}

\begin{corollary}\corlab{main:structure}
\lemref{fuzzy:ridge} implies that, given $q_1,\dots, q_k$, a representative set of
proxies for the imprecise minima of a regular imprecise terrain, it holds that
\begin{eqnarray*}
\bigcup_i (\PoWS(q_i) \setminus \PsWS(q_i)) &=&
\bigcup_i \left(\PoWS(q_i) \setminus \compl{\bigcup_{j\neq i} \PotentialWS(q_j)}\right)\\
&=& \bigcup_i \left(\PoWS(q_i)\cap \bigcup_{j\neq i} \PotentialWS(q_j)\right)\\
&=& \bigcup_i \bigcup_{j\neq i} \left(\PoWS(q_i) \cap \PoWS(q_j)\right).
\end{eqnarray*}
By \obsref{flatbottom}, this is equal to the fuzzy ridge of this terrain as
defined in \defref{fuzzy:ridge}.
This relationship is illustrated in \figref{fuzzy:ridge}.
\end{corollary}

\begin{figure}[htb]\center
\includegraphics{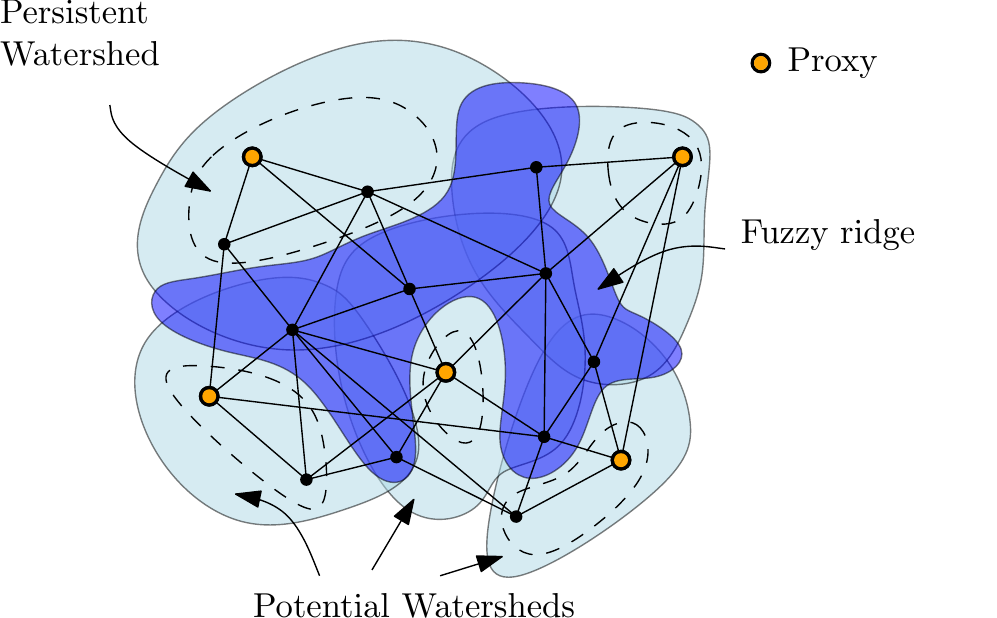}
\caption{Illustration to the fuzzy ridge on a regular terrain.}
\figlab{fuzzy:ridge}
\end{figure}

Combining this with \lemref{findtheproxies} and \lemref{all:fuzzy:boundaries} we obtain:

\begin{theorem} \thmlab{compute:fuzzy:ridge}
We can compute the fuzzy ridge of a regular imprecise terrain in
$O(n \log n)$ time, where $n$ is the number of edges of the imprecise terrain.
\end{theorem}

Note that \defref{fuzzy:ridge} can also be applied to non-regular terrains, since it is solely based on the potential watersheds of the imprecise minima. We can use the algorithm of \secref{regularize} to compute proxies for these minima, and
then use the algorithm of \lemref{all:fuzzy:boundaries} to compute a fuzzy ridge between the watersheds of these proxies efficiently for non-regular terrains. However, note that the result may not be exactly the same as the fuzzy ridge according to \defref{fuzzy:ridge}, because on a non-regular terrain, the potential watersheds of the proxies may be smaller than the potential watersheds of the imprecise minima.

%
%

\section{Conclusions}
\seclab{conclusions}


In this paper we studied flow computations on imprecise terrains under two
general models of water flow.
For the surface model, where flow paths are traced across the surface of an
imprecise polyhedral terrain, we showed NP-hardness for deciding whether water can flow
between two points.  For the network model, where flow paths are traced along
the edges of an imprecise graph, we gave efficient algorithms to compute potential
(maximal) and persistent (minimal) watersheds and potential downstream areas.
Our algorithms also work for sets of nodes and can therefore be applied to
reason about watersheds of areas, such as lakes and river beds.

In order to enable several extensions to these results in the network model, we
introduced a certain class of imprecise terrains, which we call regular.  We
first defined when a set of vertices on an imprecise terrain can be considered a
'stable' imprecise minimum.  We then described how to turn a non-regular terrain
into a regular terrain using an algorithm by Gray et al.~\cite{gkls-rleit-10}
and showed that this regularization algorithm preserves
these imprecise minima. Interestingly, this algorithm also minimizes the number
of minima of the terrain, while respecting the elevation bounds, as shown in
\cite{gkls-rleit-10}.

We showed that persistent watersheds are nested on regular terrains and that
these terrains have a fuzzy ridge structure which delineates the persistent
watersheds of these stable minima. We gave an algorithm to compute this
structure in $O(n \log n)$ time, where $n$ is the number of edges of the terrain.
The correspondence between the imprecise minima of the regular and the
non-regular terrain suggests that this fuzzy watershed decomposition on the
regular terrain also allows us to reason about the structure of the watersheds
on the original non-regular terrain.
We think that, even though our work is motivated by geographical
applications, the results will be useful in other application areas where
watersheds are being computed, for instance in image segmentation
\cite{dcj-cmmis-00}.


There are many open problems for further research.
Clearly, the contrast between the results in the surface model vs.\ the results in
the network model leaves room for further questions, e.g.,
can we develop a model to measure the quality of approximations of water flow in the surface model, and how
does it relate to the network model?
Other flow models have been proposed in the \textsc{gis} literature, e.g. \mbox{D-$\infty$},
in which the incoming water at a vertex is distributed among
the outgoing descent edges according to steepness. These models can be seen as
modified network models which approximate the steepest descent direction more
truthfully.  In order to apply the techniques we developed for watersheds, we
first need to formalize \emph{to which extent} a node is part of a watershed in
these models.

\paragraph{Acknowledgments.}
{\small We thank Chris Gray for many interesting
and useful discussions on the topic of this paper.}



\bgroup
\small
\bibliographystyle{abbrv}
\bibliography{fuzzy}
\egroup

\break
\appendix

\section{Computing potential watersheds in linear time}
\label{sec:lineartime}

\setcounter{theorem}{2}
\begin{theorem}
The canonical realization of the potential watershed 
of a set of cells $Q$ in an
imprecise grid terrain of $n$ cells can be computed in $O(n)$ time.
\end{theorem}
\begin{proof}
The computation of potential watersheds in \secref{potential-ws} has much in common with computing single-source shortest paths. In both cases, the goal is to compute a label $\delta(v)$ for each node~$v$: in the case of potential watersheds it is the lowest elevation such that a flow path to a given destination~$q$ exists; in the case of shortest paths it is the distance from the given source~$q$. During the computation, we maintain \emph{tentative} labels $d[v]$ for each node $v$ which are upper bounds on the labels to be computed. (The tentative label of a node that has not been discovered yet would be $\infty$.) The computations consist of a sequence of edge \emph{relaxations}: when relaxing a directed edge $(u,v)$, we try to improve (that is, lower) $d[v]$ based on the current value of $d[u]$, which is an upper bound on $\delta(u)$. Both problems share some crucial properties: for every node~$v$ that can be reached, there is a ``shortest'' path $\pi(v) = u_0,u_1,...,u_k$ where $u_0 = q$ and $u_k = v$, the correct labels $\delta(u_0),\delta(u_1),...,\delta(u_k)$ form a non-decreasing sequence, and when the edges on this path are relaxed in order from $(u_0,u_1)$ to $(u_{k-1},u_k)$, the relaxation of $(u_{i-1},u_i)$ will correctly set $d[u_i]$ equal to $\delta(u_i)$. All that is necessary for the computations to compute all labels, is that the sequence $\rho$ of relaxations performed by the algorithm contains $\pi(v)$ as a subsequence, for each $v$. Note that the edges of $\pi(v)$ do not need to be consecutive in $\rho$: the labels along $\pi(v)$ are computed correctly even if the relaxations of $\pi(v)$ are interleaved with relaxations of other edges, or even with out-of-order relaxations of edges of $\pi(v)$.

There are several algorithms to find a sequence of relaxations $\rho$ in the above setting, such that for every node $v$, the sequence $\rho$ contains the relaxations of a shortest path $\pi(v)$ as a subsequence. These algorithms are usually known as algorithms to compute (single-source) shortest paths, but they can also be applied directly to the more general setting described above. Dijkstra's algorithm finds a sequence of relaxations that is optimal in the sense that it relaxes each edge only once. However, to achieve this, the algorithm needs $\Theta(n)$ operations on a priority queue of size $\Theta(n)$ in the worst case, where $n$ is the number of nodes and edges in the graph~\cite{clrs-ia-09}.

An alternative is the algorithm of Henzinger et al.~\cite{hkrs-fspapg-97}. This algorithm uses a hierarchy of priority queues. Most priority queue operations in this algorithm are on small priority queues. The algorithm needs more relaxations than Dijkstra's algorithm, but still not more than $O(n)$. Provided the relaxations take constant time each, the whole algorithm runs in $O(n)$ time. However, the algorithm by Henzinger et al.\ only works if a recursive decomposition of the graph is provided that satisfies certain properties. Fortunately such decompositions can be found in $O(n)$ time for planar graphs, and also for certain other types of graphs. In particular, it is easy to construct such a decomposition for a graph that represents a grid terrain model, even in the model where each cell can drain to one or more of its eight neighbors, for which the adjacency graph is non-planar. Let $r_1 < r_2 < ... $ be a sequence of powers of four. Now we can easily make a decomposition of the graph into square regions of $\sqrt{r_1} \times \sqrt{r_1}$ nodes; we group these together into regions of $\sqrt{r_2} \times \sqrt{r_2}$ regions, etc., generally grouping regions of $\sqrt{r_i} \times \sqrt{r_i}$ nodes into regions of $\sqrt{r_{i+1}} \times \sqrt{r_{i+1}}$ nodes (some regions at the boundary of the whole input grid may be slightly smaller). On each level $i$, the regions have size $\Theta(r_i)$ and each region has $\Theta(\sqrt{r_{i}})$ nodes on its boundary, thus each level forms a so-called $r_i$-division. We choose the region sizes such that they satisfy Equation~(19) from Henzinger et al.

With this decomposition, the structure of the single-source shortest paths algorithm from Henzinger et al.\ can also be applied to the computation of potential watersheds on grid terrains. For grid terrains, $d_{\max} = O(1)$, and thus, the computation of the slope diagrams and the $O(n)$ relaxation steps from the ``shortest-paths'' algorithm take only $O(n)$ time. Together with $O(n)$ time for priority queue operations, we get a total running time of $O(n)$.
\end{proof}

\section{Persistent watersheds with multiple connected components}

\begin{lemma}\lemlab{disconnected-per-ws}
There exists a regular terrain that contains a persistent watershed that consists of more than one connected component.
\end{lemma}

\begin{proof}
Refer to \figref{disconnected-persistent-watershed}.
 The figure shows five nodes with their elevation intervals. The edges $(a,b)$,
 $(b,d)$ and $(c,d)$ have length 1. The edge $(d,e)$ has length $1.6$. From $a$
 and $e$, very steep edges lead to nodes downwards not shown in the figure. The
 potential watershed $\PoWS(e)$ of $e$ is $\{c,d,e\}$. The node $d$ is not in
 the persistent watershed of $e$: if $d$ has elevation more than $6\frac13$, the
 flow path from $d$ will lead to $b$, outside $\PoWS(e)$. In that case $c$ is a
 local minimum inside $\PoWS(e)$. Whenever $c$ is not a local minimum, the
 elevation of $d$ must be less than 4, and the flow path from $c$ will lead to
 $d$ and on to $e$. Thus $c$ is in the persistent watershed $\PsWS(e)$ of $e$,
 but $d$ is not, so we have $\PsWS(e) = \{c, e\}$.  \end{proof}

\begin{figure}[htb]\center
\includegraphics[width=0.45\textwidth]{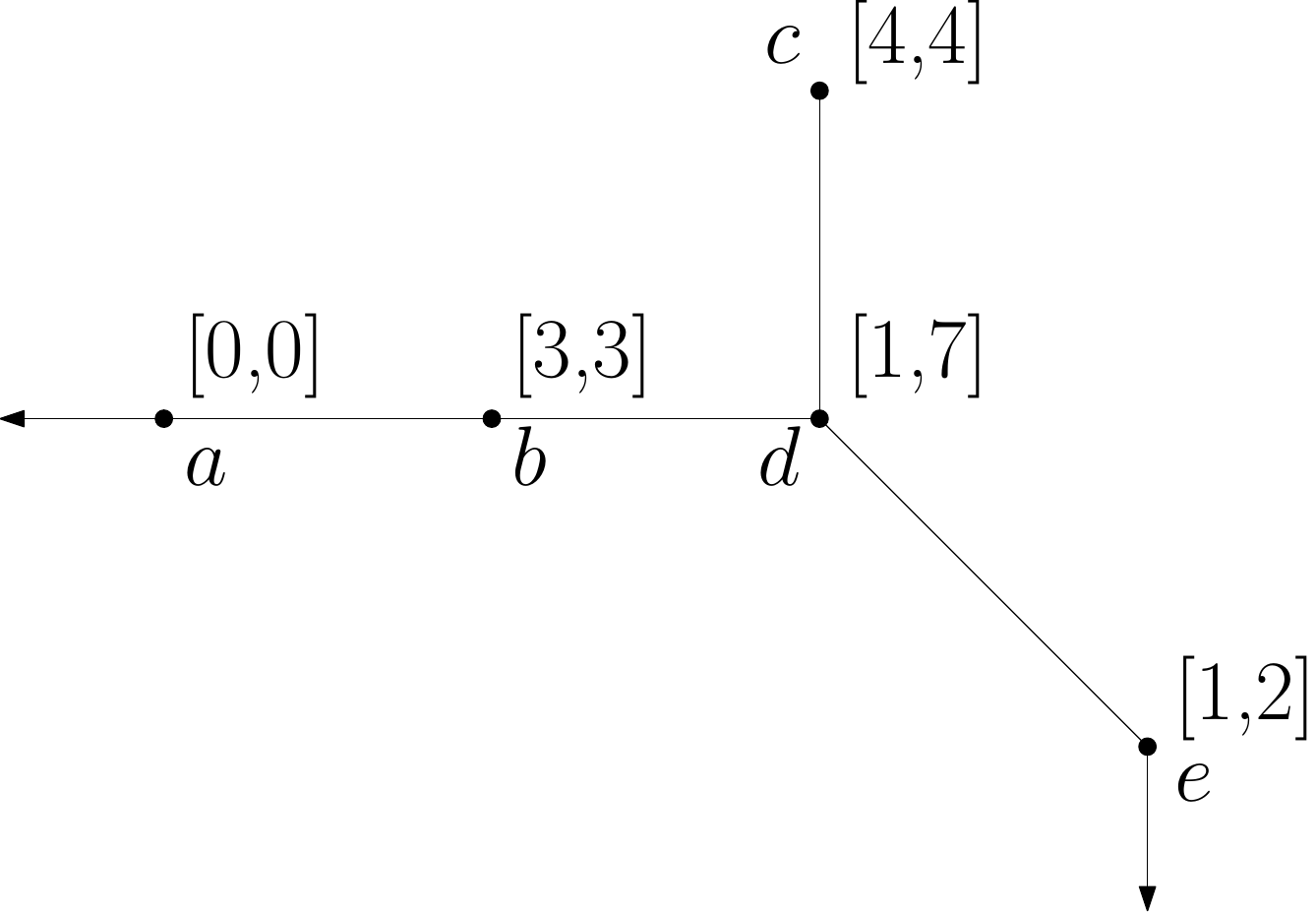}
\caption{Example of disconnected persistent watershed on a regular terrain.}
\figlab{disconnected-persistent-watershed}
\end{figure}

\end{document}